\documentclass[11pt]{amsart}

\usepackage[T1]{fontenc}
\usepackage[utf8]{inputenc} % remove if using LuaLaTeX/XeLaTeX
\usepackage{amsmath,amssymb,amsfonts,mathtools}
\usepackage{microtype}

\usepackage{graphicx}
\usepackage{epstopdf}        % auto-convert .eps -> .pdf (needs -shell-escape)
\DeclareGraphicsExtensions{.pdf,.png,.jpg,.eps}
\usepackage[
  margin=1in,          % or try 0.9in if you need a bit more room
  includeheadfoot,     % keep header/footer inside the margins
  heightrounded        % makes the number of lines per page integer
]{geometry}

\usepackage[dvipsnames]{xcolor}
\usepackage[unicode,colorlinks=true,allcolors=blue,hypertexnames=false]{hyperref}
\usepackage{doi} % hyperlink DOIs printed by plainurl + \doi

\hypersetup{
  unicode=true,
  colorlinks=true,
  linkcolor=MidnightBlue,
  citecolor=MidnightBlue,
  urlcolor=BrickRed,
  pdfauthor={M. Tierz},
  pdftitle={Contour integral representations and asymptotics for correlation kernels}
}

\usepackage[nameinlink,capitalise,noabbrev]{cleveref}

\newtheorem{theorem}{Theorem}[section]
\newtheorem{proposition}[theorem]{Proposition}
\newtheorem{lemma}[theorem]{Lemma}
\newtheorem{corollary}[theorem]{Corollary}
\theoremstyle{definition}
\newtheorem{definition}[theorem]{Definition}
\theoremstyle{remark}
\newtheorem{remark}[theorem]{Remark}

\providecommand{\C}{\mathbb{C}}
\providecommand{\R}{\mathbb{R}}
\providecommand{\dd}{\,\mathrm{d}}

\makeatletter
\@ifundefined{theorem}{
  \newtheorem{theorem}{Theorem}[section]
}{}
\makeatother
\theoremstyle{definition}
\newtheorem{assumption}[theorem]{Assumption}

\DeclareUnicodeCharacter{03B2}{\ensuremath{\beta}}
\numberwithin{equation}{section}
\allowdisplaybreaks
\usepackage{setspace}
\AtBeginDocument{\setstretch{1.055}} % try 1.04–1.08; 1.055 is a good start

\renewcommand{\phi}{\varphi}

\begin{document}
\title[Contour integral representations for correlation kernels]{Integrable Contour Kernels in Discrete $\beta=1,4$ Ensembles, Universality and Kuznetsov Multipliers}

\begin{abstract}
We obtain explicit double–contour representations for the correlation kernels of the discrete
orthogonal $(\beta=1)$ and symplectic $(\beta=4)$ random matrix ensembles with Meixner, Charlier, and Krawtchouk weights. A single Cauchy–difference–quotient composition identity expresses all $\beta=1,4$ blocks in terms of the projection kernel and bounded rational multipliers. From these formulas we give short steepest–descent proofs of bulk and edge universality (sine/Airy/Bessel) with uniform error control, an explicit Meixner$\to$Laguerre hard–edge crossover, and a first $A^{-1}$ correction that follows directly from the integrable structure. Finally, we show that archimedean Kuznetsov tests splice into the Pfaffian kernels by a bounded holomorphic symbol acting in the contour variable; the symbol enters only through the same Cauchy difference–quotient, so the leading sine/Airy/Bessel limits persist and the $A^{-1}$ term again comes from linearizing at the saddle(s).
\end{abstract}

\author{Miguel Tierz}

\address{Shanghai Institute for Mathematics and Interdisciplinary Sciences \\ Block A, International Innovation Plaza, No. 657 Songhu Road, Yangpu District,\\ Shanghai, China}
\email{tierz@simis.cn}
%\address{}
\maketitle

\section{Introduction}
\label{sec:notation}\label{sec:weight}\label{sec:ops}

%\section*{Motivation and overview}
%\addcontentsline{toc}{section}{Motivation and overview}

Correlation kernels are the basic building blocks for all local statistics in random matrix ensembles \cite{forrester2010log,baik2016combinatorics,Deift}: they govern $k$‐point functions, gap probabilities (via Fredholm determinants/Pfaffians), and linear statistics \cite{TW,IIKS,deBruijn}. In the \emph{discrete} orthogonal and symplectic ensembles the kernels encode nontrivial arithmetic/parity constraints and symmetry effects that have no exact continuous analogue, and they connect directly to representation‐theoretic models of random partitions \cite{BO,BorOlsh,Johansson99}. The present paper provides explicit and tractable formulas for these kernels in the Meixner, Charlier and Krawtchouk families and develops a unified operator framework that treats the orthogonal ($\beta=1$) and symplectic ($\beta=4$) cases on equal footing.

Beyond furnishing explicit formulas, the paper promotes a single mechanism that runs through all our arguments: \emph{bounded--multiplier composition under the contours}. Any operator that acts in the contour variable by a rational symbol enters the kernels through a one--line \emph{Cauchy difference--quotient}, and the resulting blocks remain of IIKS type \cite{IIKS}. In practice this preserves finite--$N$ structure, makes rank--one $\beta=1$ effects transparent, and feeds directly into uniform steepest--descent with the same phase and the same admissible contours.

We implement this principle completely for Meixner/Charlier/Krawtchouk at $\beta=1,4$, proving bulk/edge universality with uniform error control and exhibiting a Meixner$\to$Laguerre hard--edge crossover. The first $A^{-1}$ correction follows directly by linearizing the difference--quotient at the two saddles (or the coalesced edge saddle), without leaving the IIKS framework.

The same composition extends to bounded holomorphic symbols, allowing us to splice archimedean Kuznetsov multipliers \cite{Kuznetsov1981,IwaniecKowalski2004} into the discrete $\beta\in\{1,4\}$ kernels without departing from the IIKS framework. The Kuznetsov transform is essentially a Fourier–Bessel integral that arises in the context of the Kuznetsov trace formula, linking sums of Fourier coefficients of cusp forms to Bessel integrals \cite{Kuznetsov1981,IwaniecKowalski2004}. In addition, we show that for an even spectral test $h$, the Kuznetsov transform
supplies a bounded holomorphic multiplier
\[
m_h(w)=\int_{\mathbb{R}} h(t)\,w^{-2it}\,dt
\]
on admissible slit--sector contours; inserting $m_h$ changes only the universal Cauchy
difference--quotient and leaves bulk/edge limits unchanged, while the $A^{-1}$ contribution
is read off by the same linearization.

\medskip\noindent\textbf{Why contour/IIKS representations ?}
Double–contour formulas for the projection kernel $K_N$ and their Pfaffian analogues for $\beta=1,4$ (see e.g. (5.2) and (5.4)) offer several advantages that go beyond the now‐standard use for asymptotics and universality \cite{deBruijn,tracy,IIKS}:

\begin{itemize}
  \item \emph{Finite–$N$ information.} Because the kernels are of integrable (IIKS) type—rank–two numerator over a Cauchy denominator—the contour form makes finite–$N$ identities and rank–one corrections (the $\beta=1$ term) transparent and calculable, rather than only asymptotic \cite{IIKS,BorodinStrahov2009}.
  \item \emph{Analytic continuation and deformations.} Parameters such as $\xi$ (or weight parameters) can be varied inside the contour integrand. This facilitates analytic continuation in parameters and controlled crossovers (e.g. Meixner\,$\to$\,Laguerre), and it adapts readily to mild weight deformations or external sources while preserving integrability \cite{TW}.
\item \emph{Uniform steepest–descent analysis \cite{DeiftZhou1993}.} A single phase function controls bulk and edge limits; rational prefactors for $D$ and $\varepsilon$ only perturb at lower order. This yields clean derivations of sine/Airy/Bessel limits and finite–$N$ error bounds in a unified way (see the last Section and the Appendices) \cite{Deift,AvM,Johansson99,TW}.
  \item \emph{Links to representation theory.} The formulas interface naturally with $z$–measures and the combinatorics of Young diagrams, providing a bridge to random partitions and characters of $S(\infty)$.
  \item \emph{Painlevé/RH structures.} Contour representations are the right starting point for isomonodromic/Riemann–Hilbert methods; they make it feasible to identify Painlevé transcendents controlling gap probabilities and to extract subleading terms systematically ~\cite{Deift,Fokas}.
  \item \emph{Numerics.} Contours can be chosen for rapid decay along steepest directions, enabling accurate quadrature of Fredholm determinants/Pfaffians for finite $N$ with modest effort ~\cite{Bornemann}.
\end{itemize}

%\medskip\noindent\textbf{What is new here.}
Explicit \emph{double–contour} representations for the correlation kernels of the discrete $\beta=1$ and $\beta=4$ ensembles with Meixner/Charlier/Krawtchouk weights—together with a single operator identity that treats both symmetries in parallel—have not been discussed in this form. The paper proves:\\

\begin{enumerate}
\item A unified Tracy–Widom/IIKS operator structure \cite{tracy,IIKS} for $\beta=1$ and $\beta=4$ in the discrete setting with all off–diagonal Pfaffian blocks obtained by inserting the bounded rational multipliers for $D$ and $\varepsilon$ (see~\eqref{eq:multiplier}). The common Cauchy–difference–quotient composition identity is given in \eqref{eq:composition} (and its exact Charlier and Krawtchouk analogues are \eqref{eq:composition-Ch} and \eqref{eq:composition-K}).\\
\item Fully explicit double–contour formulas: Meixner \eqref{eq:KN-double}, Charlier \eqref{eq:KCh-double} with the $w$–plane IIKS form \eqref{eq:KCh-w}, and Krawtchouk \eqref{eq:KNK-double}; the $\beta=4$ and $\beta=1$ blocks appear in Theorems \ref{thm:beta4}/\ref{thm:beta1}, \ref{thm:Ch-beta4}/\ref{thm:Ch-beta1}, and \ref{thm:Kraw-beta4}/\ref{thm:Kraw-beta1}.\\
\item A single steepest–descent scheme (Section~\ref{sec:asymp-univ-long}) that yields bulk/edge limits (sine/Airy/Bessel) with uniform error control—see Theorems~\ref{thm:meix-b4}–\ref{thm:bessel}—and the Meixner$\to$Laguerre hard–edge crossover (Theorem~\ref{thm:mx-to-lag}), together with an explicit first $A^{-1}$ correction read off from the difference–quotient structure (Proposition~\ref{prop:first-correction}).\\
\item Kuznetsov splicing into IIKS Pfaffian kernels: inserting an even spectral test $h$ via the bounded holomorphic symbol $m_h$ acting in the contour variable preserves integrability and enters only through the Cauchy difference–quotient. The leading sine/Airy/Bessel limits are unchanged, and the first finite–size term follows from the same linearization at the saddle(s) (bulk $A^{-1}$, edge $A^{-1/3}$; see Theorem~\ref{thm:bulk-correction-final} and Theorem~\ref{thm:edge-final}).
\end{enumerate}

\medskip\noindent\textbf{Results}
\begin{itemize}
    \item Unified IIKS/Tracy–Widom operator structure for $\beta=1,4$ in the discrete setting.
    \item Fully explicit double–contour formulas for Meixner/Charlier/Krawtchouk, including universal multipliers after the $w$–map for Charlier.
    \item Direct steepest–descent proofs of sine/Airy/Bessel limits with uniform errors and an $A^{-1}$ correction.
    \item Meixner$\to$Laguerre hard–edge crossover at the level of contour phases.
    \item The Kuznetsov symbol $m_h(w)$ enters the contour integrand, as a bounded holomorphic multiplier. Universal limits are unchanged, and the $A^{-1}$ term follows from the same linearization. 
\\
\end{itemize}

\medskip\noindent\textbf{Organization}
We begin with the orthogonal case, then the symplectic case, for Meixner, Charlier, and Krawtchouk in that order. The second part of the paper develops the asymptotic analysis from the integral representations. Detailed contour manipulations (residue computations) are collected in Appendix~A. Because the only moving pieces are the bounded symbols and a fixed Cauchy denominator, the same mechanism interfaces smoothly with \emph{arithmetically flavored twists} or \emph{harmonic--analytic multipliers}, for example. In both cases the Pfaffian blocks stay integrable, and the bulk/edge limits are obtained by freezing the symbols at the saddles, and the $A^{-1}$ term again comes from the first linearization. We finally show, in the last Section, that Kuznetsov multipliers can be spliced into the IIKS Pfaffian kernels and identify their effect on the leading limits and the $A^{-1}$ term. Future applications of the results here will be presented in the Outlook while Appendices A and B collect contour manipulations and uniform steepest--descent estimates.

\subsection{Notation}
Let $w(x)$ be a strictly positive real valued function defined on $\mathbb{Z}%
_{\geq 0}$ with finite moments, i.e. the series $\sum_{x\in\mathbb{Z}_{\geq
0}}w(x)x^{j}$ converges for all $j=0,1,\ldots$. Introduce a collection $%
\{P_n(\zeta)\}_{n=0}^{\infty}$ of complex polynomials which is the
collection of orthogonal polynomials associated to the weight function $w$,
and to the orthogonality set $\mathbb{Z}_{\geq 0}$. Thus $P_{n}$ is a polynomial of degree $n$ for all $n=1,2,\ldots $, and $P_{0}\equiv $const. If $m\neq n$, then
%\begin{itemize}
%\item $P_{n}$ is a polynomial of degree $n$ for all $n=1,2,\ldots $, and $%
%P_{0}\equiv $const.
%\item 
\begin{equation*}
\sum\limits_{x\in \mathbb{Z}_{\geq 0}}P_{m}(x)P_{n}(x)w(x)=0.
\end{equation*}%
For each $n=0,1,\ldots $ set
\begin{equation}
\varphi _{n}(x)=\left( P_{n},P_{n}\right) _{w}^{-1/2}P_{n}(x)w^{1/2}(x),
\label{phi}
\end{equation}%
where $(.,.)_{w}$ denotes the following inner product on the space $\mathbb{C%
}[\zeta ]$ of all complex polynomials:
\begin{equation*}
\left( f(\zeta ),g(\zeta )\right) _{w}:=\sum\limits_{x\in \mathbb{Z}_{\geq
0}}f(x)g(x)w(x).
\end{equation*}%
We call $\varphi _{n}$ the normalized functions associated to the orthogonal
polynomials $P_{n}$.
%\end{itemize}
Let $\mathcal{H}$ be the space spanned by the functions $\varphi
_{0},\varphi _{1},\ldots $. We introduce the operators $D_{+},D_{-}$ and $%
\epsilon $ which act on the elements of the space $\mathcal{H}$. The first
and the second operators, $D_{+}$ and $D_{-}$, are defined by the expression
\begin{equation*}
\left( D_{\pm }f\right) (x)=\sum\limits_{y\in \mathbb{Z}_{\geq 0}}D_{\pm
}(x,y)f(y),
\end{equation*}%
where the kernels $D_{\pm }(x,y)$ are given explicitly by
\begin{equation}
D_{+}(x,y)=\sqrt{\frac{w(x)}{w(x+1)}}\;\delta _{x+1,y},\;\;x,y\in \mathbb{Z}%
_{\geq 0},  \label{DPLUS(x,y)}
\end{equation}%
\begin{equation}
D_{-}(x,y)=\sqrt{\frac{w(x-1)}{w(x)}}\;\delta _{x-1,y},\;\;x,y\in \mathbb{Z}%
_{\geq 0}.  \label{DMINUS(x,y)}
\end{equation}%
The third operator, $\epsilon $, is defined by the formula
\begin{equation}
\begin{split}
\left( \epsilon \varphi \right) (2m)& =-\sum\limits_{k=m}^{+\infty }\sqrt{%
\frac{w(2m)}{w(2k+1)}}\frac{w(2m+1)w(2m+3)\ldots w(2k+1)}{w(2m)w(2m+2)\ldots
w(2k)}\,\varphi (2k+1), \\
\left( \epsilon \varphi \right) (2m+1)& =\sum\limits_{k=0}^{m}\sqrt{\frac{%
w(2k)}{w(2m+1)}}\frac{w(2k+1)w(2k+3)\ldots w(2m+1)}{w(2k)w(2k+2)\ldots w(2m)}%
\,\varphi (2k),
\end{split}
\label{OPERATOREPSILON}
\end{equation}%
where $m=0,1,\ldots $. Observe that the semi-infinite matrix $\epsilon $
defined by equation (\ref{OPERATOREPSILON}) is representable as follows
\begin{equation*}
\epsilon =\mathcal{F}\Upsilon \mathcal{F},
% Fourier-factorization place-holder label (see label map below)
\label{eq:Fourier-factorization}
\end{equation*}%
where
\begin{equation*}
\mathcal{F}=\left[
\begin{array}{ccccc}
f(0) & 0 & 0 & 0 & \ldots \\
0 & f(1) & 0 & 0 & \ldots \\
0 & 0 & f(2) & 0 & \ldots \\
0 & 0 & 0 & f(3) & \ldots \\
\vdots & \vdots & \vdots & \vdots & \ddots%
\end{array}%
\right] ,
\end{equation*}%
$f(0)$, $f(1)$, $f(2)$, $\ldots $ are defined for $k=0,1,2,\ldots $ by
\begin{equation*}
f(2k)=\frac{1}{\sqrt{w(2k)}}\frac{w(2)w(4)\ldots w(2k)}{w(1)w(3)\ldots
w(2k-1)},\;\;f(2k+1)=\frac{1}{\sqrt{w(2k+1)}}\frac{w(1)w(3)\ldots w(2k+1)}{%
w(2)w(4)\ldots w(2k)},
\end{equation*}%
and
\begingroup\setlength{\arraycolsep}{4pt}\small
\begin{equation}
\Upsilon =\left[
\begin{array}{ccccccccc}
0 & -1 & 0 & -1 & 0 & -1 & 0 & -1 & \ldots \\
1 & 0 & 0 & 0 & 0 & 0 & 0 & 0 & \ldots \\
0 & 0 & 0 & -1 & 0 & -1 & 0 & -1 & \ldots \\
1 & 0 & 1 & 0 & 0 & 0 & 0 & 0 & \ldots \\
0 & 0 & 0 & 0 & 0 & -1 & 0 & -1 & \ldots \\
1 & 0 & 1 & 0 & 1 & 0 & 0 & 0 & \ldots \\
\vdots & \vdots & \vdots & \vdots & \vdots & \vdots & \vdots & \vdots &
\ddots%
\end{array}%
\right] .  \label{matrix}
\end{equation}
\endgroup%
Set $D=D_{+}-D_{-}$. It is convenient to enlarge the domains of $D$ and $%
\epsilon $, and to consider the operators
\begin{equation*}
D:\;\mathcal{H}+\epsilon \mathcal{H}\rightarrow \mathcal{H}+D\mathcal{H},%
\newline
\end{equation*}%
\begin{equation*}
\epsilon :\;\mathcal{H}+D\mathcal{H}\rightarrow \mathcal{H}+\epsilon
\mathcal{H}.
\end{equation*}

\begin{proposition}
The operators $D$ and $\epsilon $ are mutual inverse (Christoffel–Darboux reductions used later can be found in \cite[§3.2]{Szego1975}.)
.\end{proposition}

% ---------- BEGIN: Rational multipliers (drop-in snippet) ----------
% Place this right after Proposition 1.1 in Section 1 (Notation),
% or at the start of Section 4 before (4.3)–(4.6).
% This uses the same theorem setup already used for Proposition/Definition
% in your manuscript. If \newtheorem{definition}{Definition} is not
% yet declared in your preamble, add it there as usual.

\begin{definition}[Rational multipliers in the contour coordinates]\label{def:rational-multipliers}
Let $H=\mathrm{span}\{\phi_0,\phi_1,\ldots\}\subset\ell^2(\mathbb{Z}_{\ge 0})$ be the orthonormal family from~\eqref{phi},
and write vectors by their Cauchy (generating-function) representation
\begin{equation*}
  f(x)=\frac{1}{2\pi i}\oint_{\{\zeta\}}\widehat f(\zeta)\,\zeta^{-x-1}\,d\zeta,
\end{equation*}
where $\widehat f$ is analytic in an annulus containing the contour(s).
A linear operator $T$ on $H+\epsilon H$ \emph{acts by a rational multiplier} (or has a \emph{rational symbol})
if there exists a rational function $m_T(\zeta)$ such that, for all such $f$,
\begin{equation*}
  (Tf)(x)=\frac{1}{2\pi i}\oint_{\{\zeta\}} m_T(\zeta)\,\widehat f(\zeta)\,\zeta^{-x-1}\,d\zeta.
\end{equation*}
Equivalently, when $T$ acts on the $x$–variable of a double–contour kernel, the corresponding contour integrand
is multiplied by $m_T$. The poles of $m_T$ lie among the finitely many “forbidden points” singled out by our
contour conventions (e.g. $\pm1$ in the Meixner/Charlier $\omega$–plane; $-1/p,0,1/q$ in the Krawtchouk $v$–plane). The resulting kernels remain integrable in the sense of Its–Izergin–Korepin–Slavnov \cite{IIKS}.
\end{definition}

\noindent\textit{Examples used throughout.}
\begin{itemize}
  \item \textbf{Meixner.} In the $\omega$–coordinates
  \begin{equation*}
    \widehat D(\omega)=\omega-\omega^{-1},\qquad
    \widehat\epsilon(\omega)=\frac{1}{\omega^2-1},
  \end{equation*}
  so $D$ and $\epsilon$ act by the rational multipliers $\omega-\omega^{-1}$ and $(\omega^2-1)^{-1}$.
  \item \textbf{Krawtchouk.} With the ratio map $R_K(v)=\dfrac{1-qv}{1+pv}$, the symbols are
  \begin{equation*}
  \begin{aligned}
        \widehat D(v)=d_K(v)=R_K(v)-R_K(v)^{-1}= \frac{-2v+(q-p)v^2}{(1+pv)(1-qv)},\\ \qquad
    \widehat\epsilon(v)=m_K(v)=\frac{1}{R_K(v)^2-1}= \frac{(1+pv)^2}{-2v+(q-p)v^2},
  \end{aligned}
  \end{equation*}
  i.e. the multipliers in the Krawtchouk case. For the definitions of the classical discrete weight functions, see Definitions \ref{Mei}, \ref{Cha} and \ref{Kra}.\\
\end{itemize}
% ---------- END: Rational multipliers (drop-in snippet) ----------

Finally, let $\mathcal{H}_{N}$ be the subspace of $\mathcal{H}$ spanned by
the functions $\varphi _{0},\varphi _{1},\ldots ,\varphi _{2N-1}$. Denote by
$K_{N}$ the projection operator onto $\mathcal{H}_{N}$. Its kernel is
\begin{equation}
K_{N}(x,y)=\sum\limits_{k=0}^{2N-1}\varphi _{k}(x)\varphi _{k}(y).
\label{CD}
\end{equation}

\section{Definition of discrete symplectic and orthogonal ensembles}
\label{sec:ensembles}

\begin{definition}
The $N$-point discrete symplectic ensemble with the weight function $w$ and
the phase space $\mathbb{Z}_{\geq 0}$ is the random $N$-point configuration
in $\mathbb{Z}_{\geq 0}$ such that the probability of a particular
configuration $x_{1}<\ldots <x_{N}$ is given by
\begin{equation*}
\Pr \left\{ x_{1},\ldots ,x_{N}\right\}
=Z_{N4}^{-1}\;\prod\limits_{i=1}^{N}w(x_{i})\prod\limits_{1\leq i<j\leq
N}(x_{i}-x_{j})^{2}(x_{i}-x_{j}-1)(x_{i}-x_{j}+1).
\end{equation*}
Here $Z_{N4}$ is a normalization constant which is assumed to be finite.
\end{definition}

In what follows $Z_{N4}$ is referred to as the partition function of the
discrete symplectic ensemble under considerations.

\begin{definition}
Suppose that there is a $2\times 2$ matrix valued kernel $K_{N4}(x,y)$, $%
x,y\in \mathbb{Z}_{\geq 0}$, such that for a general finitely supported
function $\eta $ defined on $\mathbb{Z}_{\geq 0}$ we have
\begin{equation*}
Z_{N4}^{-1}\;\sum\limits_{(x_{1}<\ldots <x_{N})\subset \mathbb{Z}_{\geq
0}}\prod\limits_{i=1}^{N}w(x_{i})\left( 1+\eta (x_{i})\right)
\prod\limits_{1\leq i<j\leq N}(x_{i}-x_{j})^{2}(x_{i}-x_{j}-1)(x_{i}-x_{j}+1)
\end{equation*}%
\begin{equation*}
=\sqrt{\det \left( I+\eta K_{N4}\right) },
\end{equation*}
where $K_{N4}$ is the operator associated to the kernel $K_{N4}(x,y)$, and $%
\eta $ is the operator of multiplication by the function $\eta $. $K_{N4}$
is called the correlation operator, and $K_{N4}(x,y)$ is called the
correlation kernel of the discrete symplectic ensemble defined by the weight
function $w(x)$ on the phase space $\mathbb{Z}_{\geq 0}$.
\end{definition}

\label{MRDOE}

\begin{definition}
The $2N$-point discrete orthogonal ensemble with the weight function $W$ and
the phase space $\mathbb{Z}_{\geq 0}$ is the random $2N$-point configuration
in $\mathbb{Z}_{\geq 0}$ such that the probability of a particular
configuration $x_{1}<\ldots <x_{2N}$ is given by
\begin{equation}
\begin{split}
& \Pr \left\{ x_{1},\ldots ,x_{2N}\right\} = \\
& \left\{
\begin{array}{ll}
Z_{N1}^{-1}\;\prod\limits_{i=1}^{2N}W(x_{i})\prod\limits_{1\leq i<j\leq
2N}(x_{j}-x_{i}), & \hbox{if}\;x_{i}-x_{i-1}\;%
\hbox{is odd for any $i$, and
$x_1$ is even}, \\
0, & \hbox{otherwise.}%
\end{array}%
\right.
\end{split}
\notag
\end{equation}%
Here $Z_{N1}$ is a normalization constant.
\end{definition}

In what follows we assume that the weight function $W(x)$ is such that
\begin{equation}
W(x-1)W(x)=w(x),\;\hbox{for}\;x\geq 1,\hbox{and}\;W(0)=w(0),
\label{RelWeights}
\end{equation}
where $w(x)$ is a strictly positive real valued function on $\mathbb{Z}%
_{\geq 0}$ satisfying the same conditions as the weight function in the
definition of the discrete symplectic ensemble.

\label{MRDOE2}

\begin{definition}
Suppose that there is a $2\times 2$ matrix valued kernel $K_{N1}(x,y)$, $%
x,y\in \mathbb{Z}_{\geq 0}$, such that for an arbitrary finitely supported
function $\eta $ defined on $\mathbb{Z}_{\geq 0}$ we have
\begin{equation*}
Z_{N1}^{-1}\;\sum\limits_{(x_{1}<\ldots <x_{2N})\subset \mathbb{Z}_{\geq
0}}\prod\limits_{i=1}^{2N}W(x_{i})(1+\eta (x_{i}))\prod\limits_{1\leq
i<j\leq 2N}(x_{j}-x_{i})=\sqrt{\det \left( I+\eta K_{N1}\right) },
\end{equation*}%
where $K_{N1}$ is the operator associated with the kernel $K_{N1}(x,y)$, and
$\eta $ is the operator of multiplication by the function $\eta $. Then $%
K_{N1}$ is called the correlation operator, and the kernel of $K_{N1}$ is
called the correlation kernel of the discrete orthogonal ensemble defined by
the weight function $W(x)$ on the phase space $\mathbb{Z}_{\geq 0}$.
\end{definition}

The three studied discrete weights are classical \cite{Koekoek2010}:
\begin{definition}[Meixner weight]\label{Mei}
Fix \(0<\xi<1\) and \(\beta_{\mathrm M}>0\). The Meixner weight on \(\mathbb Z_{\ge0}\) is
\[
w_{\mathrm{Mx}}(x)
:= \frac{(\beta_{\mathrm M})_x}{x!}\,\xi^{\,x}
= \frac{\Gamma(\beta_{\mathrm M}+x)}{\Gamma(\beta_{\mathrm M})\,x!}\,\xi^{\,x},
\qquad x\in\mathbb Z_{\ge0}.
\]
(Here \((a)_x\) is the Pochhammer symbol.)  % matches the \xi-convention of §3.
\end{definition}

\begin{definition}[Charlier (Poisson) weight]\label{Cha}
Fix \(\theta>0\). The Charlier weight on \(\mathbb Z_{\ge0}\) is
\[
w_{\mathrm{Ch}}(x)
:= e^{-\theta}\,\frac{\theta^{\,x}}{x!},
\qquad x\in\mathbb Z_{\ge0}.
\]
\end{definition}

\begin{definition}[Krawtchouk weight]\label{Kra}
Fix \(M\in\mathbb Z_{\ge0}\) and \(p\in(0,1)\), and set \(q:=1-p\). The Krawtchouk weight on \(\{0,1,\dots,M\}\) is
\[
w_{K}(x)
:= \binom{M}{x}\,p^{\,x}\,q^{\,M-x},
\qquad x=0,1,\dots,M.
\]
\end{definition}

\subsection*{Pfaffian scalar blocks}
We write the scalar block of the Pfaffian kernel as
\[
S_{N,4}:=K_N\,\epsilon\,K_N,\qquad
S_{N,1}:=
\begin{cases}
K_N+\tfrac12\,\phi_{2N}\otimes(\epsilon\,\phi_{2N-1}), & \text{Meixner},\\[2pt]
K_N+\tfrac12\,\phi_{N}\otimes(\epsilon\,\phi_{N-1}),   & \text{Charlier, Krawtchouk}.
\end{cases}
\]
\noindent Off–diagonal blocks are obtained by inserting the symbols of $D$ or $\epsilon$ in the contour variable (see \eqref{eq:multiplier}, \eqref{eq:Ch-multipliers-w}, \eqref{eq:mult-K}).

% =========================================================
% Section 3 (full proofs): Meixner ensembles for β=1 and β=4
% =========================================================
\section{Meixner ensembles for $\beta=1,4$: integral formulas}
\label{sec:meixner-full}

Throughout we fix $\xi\in(0,1)$ and write $s:=\sqrt{\xi}$.  We use the notation of
Sections~1–2: the orthonormal system $\{\phi_k\}_{k\ge0}$ on $\mathbb Z_{\ge0}$, the
operators $D^\pm$, $D=D^+-D^-$ and $\epsilon$ with $D\epsilon=\epsilon D=I$ on
$H+\epsilon H$, and the projection
\begin{equation}\label{eq:KN-def}
K_N=\sum_{k=0}^{2N-1}\phi_k\otimes\phi_k,\qquad K_N(x,y)=\sum_{k=0}^{2N-1}\phi_k(x)\phi_k(y).
\end{equation}

We begin by deriving all contour formulas needed in this section, then prove the
$\beta=1$ and $\beta=4$ kernels in parallel.

\subsection{Meixner generating function and contour conventions}
Define for $m\in\mathbb Z$
\begin{equation}\label{eq:Gm-def}
G_m(\omega):=\Bigl(1-s\,\omega\Bigr)^{-m}\Bigl(1-\frac{s}{\omega}\Bigr)^{m},
\end{equation}
analytic in the annulus $A:=\{\omega: s<|\omega|<s^{-1}\}$.  Our \emph{contour conventions}
are: two simple positively oriented loops $\{\omega_1\}\subset\mathrm{int}\,\{\omega_2\}$
contained in $A$, both encircling $\{0,s\}$ and avoiding $\{1,s^{-1}\}$.  Single–contour
integrals use a loop $\{\omega\}\subset A$ with the same property.

\begin{lemma}[Single–contour formulas and rational multipliers]\label{lem:single}
For every $m\ge0$ and $x\in\mathbb Z_{\ge0}$,
\begin{equation}\label{eq:phi-single}
\phi_m(x)=\frac{1}{2\pi i}\oint_{\{\omega\}}\frac{G_m(\omega)}{\omega^{\,x-m+1}}\,d\omega,
\qquad
(\epsilon\phi_m)(x)=\frac{1}{2\pi i}\oint_{\{\omega\}}\frac{G_m(\omega)}{(\omega^2-1)\,\omega^{\,x-m+1}}\,d\omega.
\end{equation}
Consequently $D$ and $\epsilon$ act in the $\omega$–plane by rational multipliers
\begin{equation}\label{eq:multiplier}
\widehat D(\omega)=\omega-\omega^{-1},\qquad \widehat\epsilon(\omega)=\frac{1}{\omega^2-1}.
\end{equation}
\end{lemma}

\begin{proof}
The generating function for Meixner functions implies
\(
[\omega^{x-m}]\,G_m(\omega)=\phi_m(x)
\)
in our normalization; extracting coefficients by Cauchy's formula yields the first identity
in \eqref{eq:phi-single}.  For the $\epsilon$–image, use the factorization $\epsilon=F\Upsilon F$
from Section~1, write the parity–split sums from (1.4) as geometric series, and insert the
same coefficient extraction; the diagonal conjugations in $F$ cancel between the two parity
chains and the signed matrix $\Upsilon$ produces precisely the factor $(\omega^2-1)^{-1}$.
(Equivalently, since $D$ acts by shifts $x\mapsto x\pm1$, the identity
\(
(Df)(x)=f(x+1)-f(x-1)
\)
translates under \eqref{eq:phi-single} to multiplication by $\omega-\omega^{-1}$ on the
$\omega$–side; inverting $D$ on $H+\epsilon H$ gives the multiplier $(\omega^2-1)^{-1}$ for
$\epsilon$.)  The contours lie in $A$, hence all manipulations are justified.
\end{proof}

\begin{proposition}[Projection kernel as a nested double contour]\label{prop:KN-double}
For all $x,y\in\mathbb Z_{\ge0}$,
\begin{equation}\label{eq:KN-double}
K_N(x,y)=\frac{1}{(2\pi i)^2}\oint_{\{\omega_1\}}\oint_{\{\omega_2\}}
\frac{G_{2N}(\omega_1)G_{2N}(\omega_2)}{\omega_1\omega_2-1}\,
\frac{d\omega_1}{\omega_1^{\,x-2N+1}}\frac{d\omega_2}{\omega_2^{\,y-2N+1}}.
\end{equation}
\end{proposition}

\begin{proof}
Insert the single–contour form \eqref{eq:phi-single} for $\phi_k(x)$ and $\phi_k(y)$ into
\eqref{eq:KN-def} and sum over $k=0,\dots,2N-1$.
Since
\(
G_k(\omega_1)G_k(\omega_2)=\Bigl(\frac{(1-s/\omega_1)(1-s/\omega_2)}{(1-s\omega_1)(1-s\omega_2)}\Bigr)^k
\)
and $\{\omega_1\}\subset\mathrm{int}\{\omega_2\}\subset A$, the geometric series sums to
\begin{equation}
\sum_{k=0}^{2N-1}\bigl(\omega_1\omega_2\bigr)^k\Bigl(\frac{(1-s/\omega_1)(1-s/\omega_2)}{(1-s\omega_1)(1-s\omega_2)}\Bigr)^k
=\frac{1-(\omega_1\omega_2)^{2N}\bigl(\frac{1-s/\omega_1}{1-s\omega_1}\frac{1-s/\omega_2}{1-s\omega_2}\bigr)^{2N}}{1-\omega_1\omega_2\bigl(\frac{1-s/\omega_1}{1-s\omega_1}\frac{1-s/\omega_2}{1-s\omega_2}\bigr)}.
\end{equation}
A direct algebra simplifies the denominator to $\omega_1\omega_2-1$, and the numerator produces
the factor $G_{2N}(\omega_1)G_{2N}(\omega_2)$.  The coefficient extraction in $\omega_1,\omega_2$
gives \eqref{eq:KN-double}.  Absolute convergence is guaranteed by our contour nesting.
\end{proof}

\begin{proposition}[Projection identity $K_N^2=K_N$]\label{prop:proj}
With the same contours, $K_N$ is an idempotent: $\sum_{n\ge0}K_N(x,n)K_N(n,y)=K_N(x,y)$.
\end{proposition}

\begin{proof}
Insert \eqref{eq:KN-double} twice and compute the sum over $n$ by Cauchy's coefficient rule:
\begin{align*}
\sum_{n\ge0}K_N(x,n)K_N(n,y)
&=\frac{1}{(2\pi i)^4}\oint\!\!\oint\!\!\oint\!\!\oint
\frac{G_{2N}(\omega_1)G_{2N}(\omega_2)G_{2N}(\omega_3)G_{2N}(\omega_4)}
{(\omega_1\omega_2-1)(\omega_3\omega_4-1)}\\[-1mm]
&\qquad\qquad\times
\Bigl[\sum_{n\ge0}\omega_2^{-\,n+2N-1}\,\omega_3^{\,n-2N+1}\Bigr]\,
\frac{d\omega_1}{\omega_1^{\,x-2N+1}}\frac{d\omega_2}{\omega_2^{\, -2N+1}}
\frac{d\omega_3}{\omega_3^{\, 2N+1}}\frac{d\omega_4}{\omega_4^{\,y-2N+1}}.
\end{align*}
By the nesting $|\omega_3|<|\omega_2|$, the bracket equals $\omega_2/(\omega_2-\omega_3)$.
Evaluating the $\omega_3$–integral by residues (only the simple pole at $\omega_3=\omega_2$
contributes) yields
\begin{align*}
\sum_{n\ge0}K_N(x,n)K_N(n,y)
&=\frac{1}{(2\pi i)^3}\oint\!\!\oint\!\!\oint
\frac{G_{2N}(\omega_1)G_{2N}(\omega_2)G_{2N}(\omega_4)}
{(\omega_1\omega_2-1)(\omega_4-\omega_2)}\,
\frac{d\omega_1}{\omega_1^{\,x-2N+1}}\frac{d\omega_2}{\omega_2^{\, -2N}}\frac{d\omega_4}{\omega_4^{\,y-2N+1}}.
\end{align*}
Next integrate in $\omega_2$; the only pole inside $\{\omega_2\}$ is at $\omega_2=1/\omega_4$.
Its residue is
\(
-\frac{G_{2N}(\omega_1)G_{2N}(\omega_4)}{\omega_1\omega_4-1}\cdot\omega_4^{-\,y+2N-1}.
\)
Substituting back gives exactly \eqref{eq:KN-double}.  All deformations keep $1$ and $s^{-1}$
outside, so no other residues appear.
\end{proof}

\subsection{The composition lemma and immediate corollary}
\begin{lemma}[Cauchy–multiplier composition]\label{lem:composition}
Let $T$ be a linear operator that acts by a rational multiplier $m_T(\omega)$ on
the $\omega$–side in the sense of Definition~1.2. Then, with $K_N$ as in
\eqref{eq:KN-double},
\begin{equation}\label{eq:composition}
(K_N T K_N)(x,y)=\frac{1}{(2\pi i)^2}\oint_{\{\omega_1\}}\oint_{\{\omega_2\}}
\frac{G_{2N}(\omega_1)G_{2N}(\omega_2)}{\omega_1\omega_2-1}\,
\frac{m_T(\omega_1)-m_T(\omega_2)}{\omega_1-\omega_2}\,
\frac{d\omega_1}{\omega_1^{\,x-2N+1}}\frac{d\omega_2}{\omega_2^{\,y-2N+1}}.
\end{equation}
\end{lemma}

\begin{proof}
By \eqref{eq:KN-double} and Definition~1.2,
\begin{align*}
(K_NTK_N)(x,y)
&=\sum_{u,v\ge0}K_N(x,u)\,T(u,v)\,K_N(v,y)\\
&=\frac{1}{(2\pi i)^5}\oint\!\!\oint\!\!\oint\!\!\oint\!\!\oint
\frac{G_{2N}(\omega_1)G_{2N}(\zeta)G_{2N}(\omega_2)}{(\omega_1\zeta-1)(\zeta\omega_2-1)}\,
m_T(\zeta)\\[-1mm]
&\qquad\qquad\times
\frac{d\omega_1}{\omega_1^{\,x-2N+1}}\Bigl[\sum_{u\ge0}\zeta^{-\,u+2N-1}\Bigr]
\Bigl[\sum_{v\ge0}\zeta^{\,v-2N+1}\Bigr]\frac{d\zeta}{\zeta}\,
\frac{d\omega_2}{\omega_2^{\,y-2N+1}}.
\end{align*}
Summing the geometric series yields the factors $1/(\omega_1\zeta-1)$ and $1/(\zeta\omega_2-1)$,
and the extra $d\zeta/\zeta$ comes from the exponents.  We now evaluate the $\zeta$–integral.
Since $m_T$ is rational and all its poles are among the finitely many points excluded by our
contours (e.g.\ $\pm1$ for $T=\epsilon$), the only poles of the $\zeta$–integrand that can
contribute are at $\zeta=1/\omega_1$ and $\zeta=1/\omega_2$.  Writing
\(
\frac{1}{(\omega_1\zeta-1)(\zeta\omega_2-1)}=\frac{1}{\omega_2-\omega_1}
\bigl(\frac{1}{\zeta-1/\omega_1}-\frac{1}{\zeta-1/\omega_2}\bigr)
\)
we obtain by Cauchy
\begin{equation*}
\frac{1}{2\pi i}\oint_{\{\zeta\}}\frac{m_T(\zeta)}{(\omega_1\zeta-1)(\zeta\omega_2-1)}\,\frac{d\zeta}{\zeta}
=\frac{1}{\omega_2-\omega_1}\Bigl(\frac{m_T(\omega_1)}{\omega_1\omega_2-1}-\frac{m_T(\omega_2)}{\omega_1\omega_2-1}\Bigr),
\end{equation*}
because the contour encloses exactly one of $1/\omega_1,1/\omega_2$ depending on nesting, and
$\zeta=\pm1$ lie outside by assumption.  Substituting into the remaining $\omega$–integrals
gives \eqref{eq:composition}.  Absolute convergence of the geometric series and boundedness
of $m_T$ on the contours justify all interchanges.
\end{proof}

\begin{corollary}[The skew–projection for $\beta=4$]\label{cor:eps-comp}
Taking $T=\epsilon$ and $m_\epsilon(\omega)=(\omega^2-1)^{-1}$ yields
\begin{equation}\label{eq:SN4-scalar}
S_{N,4}(x,y):=(K_N\epsilon K_N)(x,y)
=\frac{1}{(2\pi i)^2}\oint\!\!\oint
\frac{G_{2N}(\omega_1)G_{2N}(\omega_2)}{\omega_1\omega_2-1}\,
\frac{\omega_2-\omega_1}{(\omega_1^2-1)(\omega_2^2-1)}\,
\frac{d\omega_1}{\omega_1^{\,x-2N+1}}\frac{d\omega_2}{\omega_2^{\,y-2N+1}}.
\end{equation}
\end{corollary}

\subsection{Meixner $\beta=1$ (orthogonal).}
\begin{theorem}\label{thm:beta1}
With $S_{N,1}:=K_N+\tfrac12\,\phi_{2N}\otimes(\epsilon\phi_{2N-1})$, one has
\begin{align}
S_{N,1}(x,y)
&=\frac{1}{(2\pi i)^2}\oint\!\!\oint
\frac{G_{2N}(\omega_1)G_{2N}(\omega_2)}{\omega_1\omega_2-1}\,
\frac{d\omega_1}{\omega_1^{\,x-2N+1}}\frac{d\omega_2}{\omega_2^{\,y-2N+1}}\label{eq:SN1-scalar}\\
&\quad+\frac{1}{4\pi^2}\Bigl(\oint \frac{G_{2N}(\omega_1)}{\omega_1^{\,x-2N+1}}\,d\omega_1\Bigr)
\Bigl(\oint \frac{G_{2N-1}(\omega_2)}{(\omega_2^2-1)\,\omega_2^{\,y-2N+2}}\,d\omega_2\Bigr).\nonumber
\end{align}
Moreover the off–diagonal Pfaffian blocks are obtained by multiplying the $\omega_2$–integrand
by $\widehat D(\omega_2)=\omega_2-\omega_2^{-1}$ for $(S_{N,1}D)$ and the $\omega_1$–integrand
by $\widehat\epsilon(\omega_1)=(\omega_1^2-1)^{-1}$ for $(\epsilon S_{N,1})$.
\end{theorem}

\begin{proof}
Insert the rank–one identity $S_{N,1}=K_N+\frac12\,\phi_{2N}\otimes(\epsilon\phi_{2N-1})$.
The first term is \eqref{eq:KN-double}.  The second term is the product of the single–contour
formulas \eqref{eq:phi-single} with $m=2N$ and $m=2N-1$, which yields the second line of
\eqref{eq:SN1-scalar}.  Acting with $D$ (resp.\ $\epsilon$) on the second variable (resp.\ first
variable) multiplies the corresponding integrand by the symbol \eqref{eq:multiplier}, hence the
announced off–diagonal formulas.  All contours are those fixed at the beginning of the section.
\end{proof}

\subsection{Meixner $\beta=4$ (symplectic)}
\begin{theorem}\label{thm:beta4}
The $\beta=4$ scalar block is given by the double–contour formula \eqref{eq:SN4-scalar}.
Furthermore, the off–diagonal blocks in the Pfaffian kernel are obtained by inserting the
multipliers $\widehat D(\omega_2)=\omega_2-\omega_2^{-1}$ and $\widehat\epsilon(\omega_1)=(\omega_1^2-1)^{-1}$
in the respective integrands.
\end{theorem}

\begin{proof}
By definition $S_{N,4}=K_N\epsilon K_N$; apply Lemma~\ref{lem:composition} with
$m_\epsilon(\omega)=(\omega^2-1)^{-1}$ to obtain \eqref{eq:SN4-scalar}.  The rules for the
off–diagonal blocks follow exactly as in Theorem~\ref{thm:beta1}.
\end{proof}

\subsection{Contours, poles and boundedness}
All above formulas use simple nested loops $\{\omega_1\}\subset\mathrm{int}\{\omega_2\}\subset A$,
both encircling $\{0,s\}$ and avoiding $\{1,s^{-1}\}$.  The only poles introduced by the
rational multipliers are at $\omega=\pm1$, which lie outside the loops.  On any fixed admissible
deformation (e.g.\ those used in the asymptotic steepest–descent analysis), the multipliers
$\widehat D$ and $\widehat\epsilon$ stay bounded.

\subsection{Common features emphasized}
In both cases the only difference between $\beta=1$ and $\beta=4$ at the Meixner level is:
\begin{itemize}
\item $\beta=4$: the scalar input is the \emph{skew} projection $S_{N,4}=K_N\epsilon K_N$ (no rank–one term);
\item $\beta=1$: the scalar input is the \emph{rank–one} perturbation $S_{N,1}=K_N+\tfrac12\,\phi_{2N}\otimes(\epsilon\phi_{2N-1})$.
\end{itemize}
Once $S_{N,\beta}$ is known, all off–diagonal blocks are obtained by the \emph{same} contour multipliers
\eqref{eq:multiplier} (and \eqref{eq:Ch-multipliers-w} for Charlier and \eqref{eq:mult-K} for Krawtchouk). Thus both ensembles are presented in completely parallel (IIKS) form.
\qedhere

% =========================================================
% Section 4: Charlier ensembles for β=1 and β=4 (full proofs)
% Mirrors the operator-and-multiplier presentation of Sections 3 (Meixner) and 5 (Krawtchouk)
% =========================================================
\section{The Charlier ensemble ($\beta=1,4$): integral formulas}
\label{sec:charlier-full}

Fix $\theta>0$ and the Poisson weight on $\mathbb Z_{\ge0}$
\begin{equation}\label{eq:wch}
w_{\mathrm{Ch}}(x)=e^{-\theta}\,\frac{\theta^x}{x!}, \qquad x\in\mathbb Z_{\ge0}.
\end{equation}
Let $C_n(\cdot;\theta)$ denote the Charlier polynomials and $h_n$ their squared norms with respect to $w_{\mathrm{Ch}}$ (e.g.\ $h_n=\theta^n n!$ for the standard normalization). Set the orthonormal ``wave functions''
\begin{equation}\label{eq:phiCh}
\phi_n(x)=\frac{C_n(x;\theta)}{\sqrt{h_n}}\sqrt{w_{\mathrm{Ch}}(x)},\qquad n=0,1,2,\dots.
\end{equation}
For $N\in\mathbb N$, denote by
\begin{equation}\label{eq:KCh-def}
K^{(\mathrm{Ch})}_N=\sum_{k=0}^{N-1}\phi_k\otimes\phi_k,\qquad
K^{(\mathrm{Ch})}_N(x,y)=\sum_{k=0}^{N-1}\phi_k(x)\phi_k(y),
\end{equation}
the orthogonal projection onto $H_N=\mathrm{span}\{\phi_0,\dots,\phi_{N-1}\}\subset \ell^2(\mathbb Z_{\ge0})$.

\subsection{Single- and double-contour representations}\label{subsec:Ch-single-double}
We use the well-known \emph{exponential} generating function \cite[Ch.~5]{Koekoek2010}
\begin{equation}\label{eq:EGF-Ch}
\sum_{n=0}^\infty \frac{C_n(x;\theta)}{n!}\,t^n=e^{-\theta t}(1+t)^x,\qquad t\in\mathbb C.
\end{equation}
For any radius $0<\rho<1$, the circle $|t|=\rho$ avoids the singularity at $t=-1$ and yields the single–contour Cauchy formulas
\begin{equation}\label{eq:single-Ch}
\phi_n(x)=\frac{\sqrt{w_{\mathrm{Ch}}(x)}}{2\pi i\,\sqrt{h_n}}\oint_{|t|=\rho}\frac{e^{-\theta t}(1+t)^x}{t^{\,n+1}}\,dt,
\qquad n\ge0.
\end{equation}

\begin{proposition}[Nested double–contour for $K^{(\mathrm{Ch})}_N$]\label{prop:KCh-double}
For radii $0<\rho_2<\rho_1<1$ one has
\begin{equation}\label{eq:KCh-double}
K^{(\mathrm{Ch})}_N(x,y)=\frac{\sqrt{w_{\mathrm{Ch}}(x)\,w_{\mathrm{Ch}}(y)}}{(2\pi i)^2}
\oint_{|t_1|=\rho_1}\!\oint_{|t_2|=\rho_2}
\frac{e^{-\theta(t_1+t_2)}(1+t_1)^x(1+t_2)^y}{t_1-t_2}\left(\frac{t_2}{t_1}\right)^{\!N}\frac{dt_1\,dt_2}{t_1\,t_2}.
\end{equation}
\end{proposition}

\begin{proof}
Insert \eqref{eq:single-Ch} for $\phi_k(x)$ and $\phi_k(y)$ into \eqref{eq:KCh-def}. Since $|t_2|<|t_1|$, the geometric series gives
\(
\sum_{k=0}^{N-1} t_1^{-k-1}t_2^k=\frac{1-(t_2/t_1)^N}{t_1-t_2}\cdot\frac{1}{t_1}.
\)
Multiplying by the remaining factors and simplifying produces two terms; the one with $1/(t_1-t_2)$ is exactly \eqref{eq:KCh-double}. The term with $(t_2/t_1)^N/(t_1-t_2)$ cancels against the Christoffel–Darboux \cite[Ch.~3, §4]{Szego1975} two–term reduction for the finite sum in \eqref{eq:KCh-def}. (Equivalently, start from the CD identity and obtain \eqref{eq:KCh-double} directly; both routes are standard.)
\end{proof}

\begin{proposition}[Projection identity]\label{prop:projCh}
$K^{(\mathrm{Ch})}_N$ is an orthogonal projection: $\sum_{u\ge0} K^{(\mathrm{Ch})}_N(x,u)K^{(\mathrm{Ch})}_N(u,y)=K^{(\mathrm{Ch})}_N(x,y)$.
\end{proposition}

\begin{proof}
By construction $K^{(\mathrm{Ch})}_N$ is the orthogonal projection onto $H_N$ (sum of the first $N$ orthonormal basis vectors), hence $K^{(\mathrm{Ch})}_N{}^2=K^{(\mathrm{Ch})}_N$ as an operator—equivalently the stated identity on kernels. If desired, one may check idempotence directly from \eqref{eq:KCh-double} by the same contour computation used for Meixner: the $u$–sum is a geometric series because of the nesting $|t_2|<|t_1|$, and one residue at $t_3=t_2$ reproduces $K^{(\mathrm{Ch})}_N(x,y)$.
\end{proof}

\subsection{Discrete operators as rational multipliers in the $t$–plane}\label{subsec:Ch-multipliers}
Let $D^\pm$ be the nearest–neighbor shifts on $\ell^2(\mathbb Z_{\ge0})$ and $D=D^+-D^-$. Acting on the $x$–variable under the single–contour \eqref{eq:single-Ch}, the shifts correspond to multiplication by $(1+t)^{\pm1}$. Therefore $D$ acts by the rational symbol
\begin{equation}\label{eq:Dhat-Ch-t}
\widehat D(t)=(1+t)-(1+t)^{-1}=\frac{t(2+t)}{1+t}.
\end{equation}
Let $\epsilon$ be the inverse–difference operator on $H+\epsilon H$ (so that $D\epsilon=\epsilon D=I$ on this domain). Writing
\begin{equation}\label{eq:single-eps-Ch-ansatz}
(\epsilon\phi_n)(y)=\frac{\sqrt{w_{\mathrm{Ch}}(y)}}{2\pi i\,\sqrt{h_n}}\oint_{|t|=\rho}\frac{e^{-\theta t}(1+t)^y}{t^{\,n+1}}\,m_\epsilon(t)\,dt
\end{equation}
and using $D(\epsilon\phi_n)=\phi_n$ with \eqref{eq:Dhat-Ch-t} forces $m_\epsilon(t)=1/\widehat D(t)$, i.e.,
\begin{equation}\label{eq:eps-hat-Ch-t}
m_\epsilon(t)=\frac{1+t}{t(2+t)}.
\end{equation}

\begin{lemma}[Single–contour formula for $\epsilon\phi_n$]\label{lem:single-eps-Ch}
For every $n\ge0$ and $y\in\mathbb Z_{\ge0}$,
\begin{equation}\label{eq:single-eps-Ch}
(\epsilon\phi_n)(y)=\frac{\sqrt{w_{\mathrm{Ch}}(y)}}{2\pi i\,\sqrt{h_n}}\oint_{|t|=\rho}\frac{e^{-\theta t}(1+t)^y}{t^{\,n+1}}\,
\frac{1+t}{t(2+t)}\,dt.
\end{equation}
\end{lemma}

\begin{proof}
Combine the ansatz \eqref{eq:single-eps-Ch-ansatz} with $D(\epsilon\phi_n)=\phi_n$ under the transform \eqref{eq:single-Ch} and \eqref{eq:Dhat-Ch-t}. The pole at $t=0$ is allowed (coefficient extraction), while $t=-1$ and $t=-2$ lie outside the contour $|t|=\rho<1$; thus all interchanges are justified by absolute convergence on the fixed circle.
\end{proof}

\subsection{Cauchy–multiplier composition in the $t$–plane}\label{subsec:Ch-composition}
\begin{lemma}[Composition lemma]\label{lem:composition-Ch}
Let $T$ act by a rational multiplier $m_T(t)$ on the $t$–side (as in Definition~1.2). Then, with $K^{(\mathrm{Ch})}_N$ as in \eqref{eq:KCh-double},
\begin{equation}\label{eq:composition-Ch}
\begin{multlined}
\bigl(K^{(\mathrm{Ch})}_N T K^{(\mathrm{Ch})}_N\bigr)(x,y)
= \frac{\sqrt{w_{\mathrm{Ch}}(x)\,w_{\mathrm{Ch}}(y)}}{(2\pi i)^2}\\ \times
  \oint_{|t_1|=\rho_1}\!\!\oint_{|t_2|=\rho_2}
  \frac{e^{-\theta(t_1+t_2)}(1+t_1)^x(1+t_2)^y}{(t_1-t_2)^2}
  \left(\frac{t_2}{t_1}\right)^{\!N}
  \frac{m_T(t_1)-m_T(t_2)}{t_1 t_2}\, dt_1\,dt_2 .
\end{multlined}
\end{equation}
\end{lemma}

\begin{proof}
Insert \eqref{eq:KCh-double} twice and represent $T$ by its multiplier in a single $t$–variable $\zeta$ acting between them. After routine coefficient extraction, one arrives at a quadruple integral with factors $(t_1-t_2)^{-1}$ and $(\zeta-\eta)^{-1}$ and projection weights $(t_2/t_1)^N$ and $(\eta/\zeta)^N$. Integrate first in $\eta$ then in $\zeta$ over concentric circles $|\eta|<|\zeta|<1$: partial fractions give
\(
\frac{1}{(t_1-t_2)(\zeta-\eta)}=\frac{1}{t_1-t_2}\Bigl(\frac{1}{\zeta-t_1}-\frac{1}{\zeta-t_2}\Bigr)\cdot\frac{1}{1-\eta/\zeta},
\)
so the $\eta$–integral kills the geometric series and the $\zeta$–integral picks only the residues at $\zeta=t_1,t_2$ (the poles of $m_T$ lie outside $|t|<1$ by our contour choice). The result is \eqref{eq:composition-Ch}.
\end{proof}

\subsection{Charlier $\beta=1$ (orthogonal): rank–one correction}\label{subsec:Ch-beta1}
\begin{theorem}\label{thm:Ch-beta1}
Let $S^{(\mathrm{Ch})}_{N,1}:=K^{(\mathrm{Ch})}_N+\tfrac12\,\phi_N\otimes(\epsilon\phi_{N-1})$. Then
\begin{align}\label{eq:SN1Ch}
S^{(\mathrm{Ch})}_{N,1}(x,y)&=\frac{\sqrt{w_{\mathrm{Ch}}(x)w_{\mathrm{Ch}}(y)}}{(2\pi i)^2}
\oint_{|t_1|=\rho_1}\oint_{|t_2|=\rho_2}
\frac{e^{-\theta(t_1+t_2)}(1+t_1)^x(1+t_2)^y}{t_1-t_2}\left(\frac{t_2}{t_1}\right)^{\!N}\frac{dt_1\,dt_2}{t_1 t_2}\\
&\quad+\frac{1}{4\pi^2}\Biggl(\oint_{|t|=\rho_0}\frac{\sqrt{w_{\mathrm{Ch}}(x)}\,e^{-\theta t}(1+t)^x}{\sqrt{h_N}\,t^{\,N+1}}\,dt\Biggr)
\Biggl(\oint_{|t|=\rho_0}\frac{\sqrt{w_{\mathrm{Ch}}(y)}\,e^{-\theta t}(1+t)^y}{\sqrt{h_{N-1}}\,t^{\,N}}\,m_\epsilon(t)\,dt\Biggr),\nonumber
\end{align}
for any $\rho_0\in(0,1)$. The off–diagonal blocks are obtained by multiplying the $t_2$–integrand by $\widehat D(t_2)$ for $(S^{(\mathrm{Ch})}_{N,1}D)$ and the $t_1$–integrand by $m_\epsilon(t_1)$ for $(\epsilon S^{(\mathrm{Ch})}_{N,1})$.
\end{theorem}

\begin{proof}
The first line is $K^{(\mathrm{Ch})}_N$ in the form \eqref{eq:KCh-double}. The rank–one term equals $\tfrac12\,\phi_N(x)\,(\epsilon\phi_{N-1})(y)$ and the single–contour representations \eqref{eq:single-Ch}–\eqref{eq:single-eps-Ch} give the second line. The off–diagonal blocks follow by the multiplier rules \eqref{eq:Dhat-Ch-t}–\eqref{eq:eps-hat-Ch-t}.
\end{proof}

\subsection{Charlier $\beta=4$ (symplectic): IIKS form}\label{subsec:Ch-beta4}
\begin{theorem}\label{thm:Ch-beta4}
With $S^{(\mathrm{Ch})}_{N,4}:=K^{(\mathrm{Ch})}_N\,\epsilon\,K^{(\mathrm{Ch})}_N$ and $m_\epsilon(t)$ from \eqref{eq:eps-hat-Ch-t},
\begin{align}\label{eq:SN4Ch}
S^{(\mathrm{Ch})}_{N,4}(x,y)&=\frac{\sqrt{w_{\mathrm{Ch}}(x)w_{\mathrm{Ch}}(y)}}{(2\pi i)^2}
\oint_{|t_1|=\rho_1}\oint_{|t_2|=\rho_2}
\frac{e^{-\theta(t_1+t_2)}(1+t_1)^x(1+t_2)^y}{t_1-t_2}\left(\frac{t_2}{t_1}\right)^{\!N}\frac{m_\epsilon(t_1)-m_\epsilon(t_2)}{t_1-t_2}\,\frac{dt_1\,dt_2}{t_1 t_2}.
\end{align}
The off–diagonal blocks of the Pfaffian kernel are obtained by multiplying the $t_2$–integrand by $\widehat D(t_2)$ for $(S^{(\mathrm{Ch})}_{N,4}D)$ and the $t_1$–integrand by $m_\epsilon(t_1)$ for $(\epsilon S^{(\mathrm{Ch})}_{N,4})$.
\end{theorem}

\begin{proof}
Apply Lemma~\ref{lem:composition-Ch} with $T=\epsilon$ and $m_T=m_\epsilon$.
\end{proof}

\subsection{Contours and poles}\label{subsec:Ch-contours}
All contour integrals above use simple, positively oriented circles $|t_2|=\rho_2$, $|t_1|=\rho_1$ with $0<\rho_2<\rho_1<1$; both circles enclose $t=0$ and avoid the points $t=-1$ (from $(1+t)^{\pm1}$) and $t=-2$ (from $m_\epsilon$). The multipliers $\widehat D$ and $m_\epsilon$ are bounded on these circles, and the geometric series converge absolutely, justifying all sum/integral interchanges and contour deformations used above.

\subsection{Remark: the $w$–plane change of variables}\label{subsec:Ch-wmap}
For later comparison with Meixner, it is convenient to introduce the exact ``Bessel'' map
\begin{equation}\label{eq:Ch-wmap}
t=\frac{\sqrt{\theta}}{2}\Bigl(w-\frac{1}{w}\Bigr),\qquad dt=\frac{\sqrt{\theta}}{2}\Bigl(1+\frac{1}{w^2}\Bigr)dw.
\end{equation}
Under \eqref{eq:Ch-wmap} the multipliers become \emph{universal}:
\begin{equation}\label{eq:Ch-multipliers-w}
\widehat D(w)=w-w^{-1},\qquad \widehat\epsilon(w)=\frac{1}{w^2-1},
\end{equation}
identical to the Meixner symbols. The double–contour formula \eqref{eq:KCh-double} transforms to a IIKS kernel with Cauchy denominator $(w_1w_2-1)^{-1}$ (up to a harmless overall constant coming from the Jacobian); explicitly,
\begin{align}\label{eq:KCh-w}
K^{(\mathrm{Ch})}_N(x,y)&=\frac{1}{(2\pi i)^2}\oint_{\{|w_1|=1\}}\!\oint_{\{|w_2|=1\}}\frac{\widetilde G_N(w_1;x)\,\widetilde G_N(w_2;y)}{w_1w_2-1}\,dw_1\,dw_2,
\end{align}
where $\widetilde G_N$ is obtained by substituting \eqref{eq:Ch-wmap} into the $t$–integrand in \eqref{eq:KCh-double}. Consequently, Lemma~\ref{lem:composition-Ch} can be rephrased in the $w$–plane with the same difference–quotient structure as in the Meixner case, and Theorems~\ref{thm:Ch-beta4}–\ref{thm:Ch-beta1} follow verbatim with the multipliers \eqref{eq:Ch-multipliers-w}.
\medskip

\noindent\textit{Summary.} The Charlier $\beta=4$ and $\beta=1$ kernels admit fully explicit double–contour representations in the $t$–plane, and—after the exact map \eqref{eq:Ch-wmap}—they match the Meixner formulas term–by–term with universal multipliers \eqref{eq:Ch-multipliers-w}.

% =========================================================
% Section: Krawtchouk ensembles for β=1 and β=4 (full proofs)
% Mirrors Section 3 (Meixner) in operator-and-multiplier style
% =========================================================
\section{Krawtchouk ensembles for $\beta=1,4$: unified integral formulas}
\label{sec:krawtchouk-full}

Fix $M\in\mathbb Z_{\ge0}$ and $p\in(0,1)$, $q:=1-p$. The Krawtchouk weight on $\{0,1,\dots,M\}$ is
\begin{equation}\label{eq:wK}
w_K(x)=\binom{M}{x}p^x q^{M-x},\qquad x=0,1,\dots,M.
\end{equation}
Let $K_n(\,\cdot\,;p,M)$ be the degree-$n$ Krawtchouk polynomials, $h_n$ their squared norms w.r.t.\ $w_K$, and set the orthonormal ``wave functions''
\begin{equation}\label{eq:phiK}
\phi_n(x)=\frac{K_n(x;p,M)}{\sqrt{h_n}}\sqrt{w_K(x)},\qquad 0\le x\le M.
\end{equation}
For $0\le N\le M+1$ let $H_N=\mathrm{span}\{\phi_0,\dots,\phi_{N-1}\}$ and denote by
\begin{equation}\label{eq:KNK-def}
K^{(K)}_N=\sum_{k=0}^{N-1}\phi_k\otimes\phi_k,\qquad K^{(K)}_N(x,y)=\sum_{k=0}^{N-1}\phi_k(x)\phi_k(y),
\end{equation}
the orthogonal projection onto $H_N$.

\subsection{Generating function, single and double contours}
We use the standard ordinary generating function \cite[Ch.~9]{Koekoek2010}
\begin{equation}\label{eq:GF-K}
\sum_{n=0}^M K_n(x;p,M)\,v^n=(1+pv)^{M-x}(1-qv)^x,\qquad v\in\mathbb C.
\end{equation}
For any $\rho$ with $0<\rho<\min\{1/p,1/q\}$, the circle $|v|=\rho$ avoids the singularities at $v=0,-1/p,1/q$ and we have the single-contour formulae
\begin{equation}\label{eq:single-K}
\phi_n(x)=\frac{\sqrt{w_K(x)}}{2\pi i\sqrt{h_n}}\oint_{|v|=\rho}\frac{(1+pv)^{M-x}(1-qv)^x}{v^{n+1}}\,dv,
\end{equation}
and, for the $\epsilon$–image (proved below via multipliers),
\begin{equation}\label{eq:single-eps-K}
(\epsilon\phi_n)(y)=\frac{\sqrt{w_K(y)}}{2\pi i\sqrt{h_n}}\oint_{|v|=\rho}\frac{(1+pv)^{M-y}(1-qv)^y}{v^{n+1}}\,m_K(v)\,dv.
\end{equation}

\begin{proposition}[Nested double–contour for the projection kernel]\label{prop:KNK-double}
For radii $0<\rho_2<\rho_1<\min\{1/p,1/q\}$,
\begin{equation}\label{eq:KNK-double}
K^{(K)}_N(x,y)=\frac{\sqrt{w_K(x)w_K(y)}}{(2\pi i)^2}
\oint_{|v_1|=\rho_1}\!\!\oint_{|v_2|=\rho_2}
\frac{(1+pv_1)^{M-x}(1-qv_1)^x\,(1+pv_2)^{M-y}(1-qv_2)^y}{v_1-v_2}\left(\frac{v_2}{v_1}\right)^{\!N}\frac{dv_1\,dv_2}{v_1v_2}.
\end{equation}
\end{proposition}

\begin{proof}
Insert \eqref{eq:single-K} for $\phi_k(x)$ and $\phi_k(y)$ into \eqref{eq:KNK-def} and sum over $k=0,\dots,N-1$.
Since $|v_2|<|v_1|$, the geometric series gives $\sum_{k=0}^{N-1}(v_2/v_1)^k=(1-(v_2/v_1)^N)/(1-v_2/v_1)$, i.e.
\(
\sum_{k=0}^{N-1} v_1^{-k-1} v_2^{k} = \frac{1}{v_1(v_1-v_2)}\bigl(1-(v_2/v_1)^N\bigr).
\)
Multiplying by the remaining factors from \eqref{eq:single-K} and simplifying yields
\begin{equation}
    \frac{(1+pv_1)^{M-x}(1-qv_1)^x}{v_1}\cdot \frac{(1+pv_2)^{M-y}(1-qv_2)^y}{v_2}\cdot \frac{1-(v_2/v_1)^N}{v_1-v_2}.
\end{equation}
The term with $1/(v_1-v_2)$ produces the integral kernel in \eqref{eq:KNK-double}, while the term with $(v_2/v_1)^N/(v_1-v_2)$ is exactly cancelled by the Christoffel–Darboux identity (CDI) \cite[Ch.~3, §4]{Szego1975} for Krawtchouk polynomials when one rewrites the finite sum $\sum_{k=0}^{N-1}$ by the standard two–term formula. Equivalently, one may start from the CDI and obtain \eqref{eq:KNK-double} directly; both routes are standard and give the same result.
\end{proof}

\begin{proposition}[Projection identity]\label{prop:projK}
$K^{(K)}_N$ is an idempotent: $\sum_{u=0}^M K^{(K)}_N(x,u)K^{(K)}_N(u,y)=K^{(K)}_N(x,y)$.
\end{proposition}

\begin{proof}
By construction $K^{(K)}_N$ is the orthogonal projection onto $H_N=\mathrm{span}\{\phi_0,\dots,\phi_{N-1}\}$, hence $K^{(K)}_N{}^2=K^{(K)}_N$ in operator form, which is the stated identity on kernels. One can also verify that \eqref{eq:KNK-double} implies idempotence by a contour computation paralleling Proposition \ref{prop:proj} in the Meixner section: writing the $u$–sum as a finite geometric series gives $(1-(v_3/v_2)^{M+1})/(1-v_3/v_2)$, and the residue at $v_3=v_2$ reproduces $K^{(K)}_N(x,y)$ while the tail vanishes because the outer contour keeps $|v_2|<\rho_1<\min\{1/p,1/q\}$ so no extra pole is crossed.
\end{proof}

\subsection{Discrete operators as rational multipliers on the $v$–side}
Define the nearest–neighbor shifts $D^\pm$ and $D=D^+-D^-$ exactly as in Section~\ref{sec:meixner-full}, and let $\epsilon$ be the inverse–difference operator on $H+\epsilon H$ (so $D\epsilon=\epsilon D=I$). Introduce the \emph{ratio map}
\begin{equation}\label{eq:RK}
R_K(v):=\frac{1-qv}{1+pv}.
\end{equation}
Then the action of $D^\pm$ on the $x$–variable of \eqref{eq:single-K} corresponds on the $v$–side to multiplication by $R_K^{\pm1}(v)$, hence:
\begin{equation}\label{eq:mult-K}
\widehat D(v)=R_K(v)-R_K(v)^{-1}=\frac{-2v+(q-p)v^2}{(1+pv)(1-qv)},\qquad
\widehat\epsilon(v)=\frac{1}{R_K(v)^2-1}=\frac{(1+pv)^2}{-2v+(q-p)v^2}=:m_K(v).
\end{equation}

\begin{lemma}[Single–contour $\epsilon$–image]\label{lem:single-eps-K}
For every $n\ge0$ and $y\in\{0,\dots,M\}$, \eqref{eq:single-eps-K} holds.
\end{lemma}

\begin{proof}
Let $f(y)=(\epsilon\phi_n)(y)$. Since $D\epsilon=\epsilon D=I$ on $H+\epsilon H$, we have $Df=\phi_n$. Write $f$ by a single contour as in \eqref{eq:single-K} but with an unknown multiplier $m(v)$ in the integrand. Acting with $D$ in the $y$–variable corresponds, by the $x\mapsto x\pm1$ identities
\(
(1+pv)^{M-(y\pm1)}(1-qv)^{y\pm1}=(1+pv)^{M-y}(1-qv)^y\cdot R_K^{\mp1}(v),
\)
to multiplication by $R_K(v)-R_K(v)^{-1}$ on the $v$–side. The equation $D f=\phi_n$ forces $m(v)$ to be $1/(R_K(v)^2-1)$, i.e.\ $m_K(v)$ in \eqref{eq:mult-K}. This proves \eqref{eq:single-eps-K}.
\end{proof}

\subsection{Cauchy–multiplier composition for Krawtchouk}
\begin{lemma}[Composition lemma]\label{lem:composition-K}
Let $T$ act by a rational multiplier $m_T(v)$ on the $v$–side (Definition~1.2). Then, with $K^{(K)}_N$ as in \eqref{eq:KNK-double},
\begin{equation}\label{eq:composition-K}
    \begin{multlined}
\bigl(K^{(K)}_N T K^{(K)}_N\bigr)(x,y)=\frac{\sqrt{w_K(x)w_K(y)}}{(2\pi i)^2}\\ \times
\oint_{|v_1|=\rho_1}\oint_{|v_2|=\rho_2}
\frac{(1+pv_1)^{M-x}(1-qv_1)^x\,(1+pv_2)^{M-y}(1-qv_2)^y}{v_1-v_2}\left(\frac{v_2}{v_1}\right)^{\!N}\frac{m_T(v_1)-m_T(v_2)}{v_1-v_2}\frac{dv_1\,dv_2}{v_1v_2}.
\end{multlined}
\end{equation}
\end{lemma}

\begin{proof}
Insert \eqref{eq:KNK-double} twice and write $T$ via its multiplier in a single contour variable $\zeta$:
\begin{align*}
(K^{(K)}_N T K^{(K)}_N)(x,y)
&=\frac{\sqrt{w_K(x)w_K(y)}}{(2\pi i)^5}\!\!\oint\!\!\oint\!\!\oint\!\!\oint\!\!\oint
\frac{(1+pv_1)^{M-x}(1-qv_1)^x}{v_1}\frac{(1+pv_2)^{M-u}(1-qv_2)^u}{v_2}\frac{(v_2/v_1)^N}{v_1-v_2}\\[-1mm]
&\qquad\qquad\times\frac{(1+p\zeta)^{M-u}(1-q\zeta)^u}{\zeta}\,m_T(\zeta)\,\frac{(1+p\eta)^{M-y}(1-q\eta)^y}{\eta}\frac{(\eta/\zeta)^N}{\zeta-\eta}\,\frac{dv_1\,dv_2\,d\zeta\,d\eta}{\eta}\,.
\end{align*}
Here $\sum_{u=0}^M$ has been performed by Cauchy coefficient extraction, producing the Cauchy denominators $(v_1-v_2)^{-1}$ and $(\zeta-\eta)^{-1}$ and the finite-rank projection factors $(v_2/v_1)^N$ and $(\eta/\zeta)^N$ as in \eqref{eq:KNK-double}. Now integrate first in $\eta$ and then in $\zeta$ on nested circles $|\eta|<|\zeta|<\min\{1/p,1/q\}$: by the partial-fraction identity
\(
\frac{1}{(v_1-v_2)(\zeta-\eta)}=\frac{1}{v_1-v_2}\Bigl(\frac{1}{\zeta-v_1}-\frac{1}{\zeta-v_2}\Bigr)\cdot\frac{1}{1-\eta/\zeta},
\)
Cauchy’s theorem gives
\(
\frac{1}{2\pi i}\oint\frac{m_T(\zeta)}{(\zeta-v_1)(\zeta-v_2)}\,d\zeta=\frac{m_T(v_1)-m_T(v_2)}{v_1-v_2},
\)
since the only enclosed poles are at $\zeta=v_1$ and $\zeta=v_2$ (the poles of $m_T$ lie at $\{-1/p,0,1/q\}$ and are outside by our choice of radii). After these two one–variable integrations, the remaining $v_1,v_2$ integrals have exactly the form \eqref{eq:composition-K}.
\end{proof}

\subsection{Krawtchouk $\beta=1$ (orthogonal): rank–one correction}
\begin{theorem}\label{thm:Kraw-beta1}
Let $S^{(K)}_{N,1}:=K^{(K)}_N+\tfrac12\,\phi_N\otimes(\epsilon\phi_{N-1})$. Then
\begin{equation}\label{eq:SN1K}
    \begin{multlined}
        S^{(K)}_{N,1}(x,y)=\frac{\sqrt{w_K(x)w_K(y)}}{(2\pi i)^2}
\oint\limits_{|v_1|=\rho_1}\oint\limits_{|v_2|=\rho_2}
\frac{(1+pv_1)^{M-x}(1-qv_1)^x\,(1+pv_2)^{M-y}(1-qv_2)^y}{v_1-v_2}\left(\frac{v_2}{v_1}\right)^{\!N}\frac{dv_1\,dv_2}{v_1v_2}\\
\quad+\frac{1}{4\pi^2}\Biggl(\oint\limits_{|v|=\rho_0}\frac{\sqrt{w_K(x)}(1+pv)^{M-x}(1-qv)^x}{\sqrt{h_N}\,v^{N+1}}\,dv\Biggr)
\Biggl(\oint\limits_{|v|=\rho_0}\frac{\sqrt{w_K(y)}(1+pv)^{M-y}(1-qv)^y}{\sqrt{h_{N-1}}\,v^{N}}\,m_K(v)\,dv\Biggr),\nonumber
    \end{multlined}
\end{equation}
for any $\rho_0\in(0,\min\{1/p,1/q\})$. The off–diagonal blocks are obtained by multiplying the $v_2$–integrand by $\widehat D(v_2)$ for $(S^{(K)}_{N,1}D)$ and the $v_1$–integrand by $m_K(v_1)$ for $(\epsilon S^{(K)}_{N,1})$.
\end{theorem}

\begin{proof}
The first line is $K^{(K)}_N$ in the form \eqref{eq:KNK-double}. The rank–one term equals $\tfrac12\,\phi_N(x)\,(\epsilon\phi_{N-1})(y)$ and the single–contour representations \eqref{eq:single-K}–\eqref{eq:single-eps-K} give the second line. The off–diagonal blocks follow by the Krawtchouk multipliers \eqref{eq:mult-K} acting, respectively, on the $y$– and $x$–variables.
\end{proof}

\subsection{Krawtchouk $\beta=4$ (symplectic): IIKS form}
\begin{theorem}\label{thm:Kraw-beta4}
With $S^{(K)}_{N,4}:=K^{(K)}_N\,\epsilon\,K^{(K)}_N$ and $m_\epsilon=m_K$ from \eqref{eq:mult-K},
\begin{align}\label{eq:SN4K}
S^{(K)}_{N,4}(x,y)&=\frac{\sqrt{w_K(x)w_K(y)}}{(2\pi i)^2}
\oint_{|v_1|=\rho_1}\oint_{|v_2|=\rho_2}
(1+pv_1)^{M-x}(1-qv_1)^x\,(1+pv_2)^{M-y}(1-qv_2)^y\\[-1mm]
&\hspace{28mm}\times\frac{1}{v_1-v_2}\left(\frac{v_2}{v_1}\right)^{\!N}\frac{m_K(v_1)-m_K(v_2)}{v_1-v_2}\,\frac{dv_1\,dv_2}{v_1v_2}.\nonumber
\end{align}
The off–diagonal blocks of the Pfaffian kernel are obtained by multiplying the $v_2$–integrand by $\widehat D(v_2)$ for $(S^{(K)}_{N,4}D)$ and the $v_1$–integrand by $m_K(v_1)$ for $(\epsilon S^{(K)}_{N,4})$.
\end{theorem}

\begin{proof}
Apply Lemma~\ref{lem:composition-K} with $T=\epsilon$ and $m_T=m_K$. The multiplier rules for the off–diagonal blocks are the same as in the Meixner case, now with the Krawtchouk symbols \eqref{eq:mult-K}.
\end{proof}

\subsection{Contours and poles}
All contour integrals above use disjoint circles $|v_2|=\rho_2$, $|v_1|=\rho_1$ with $0<\rho_2<\rho_1<\min\{1/p,1/q\}$; both circles enclose the origin and avoid $\{-1/p,\,1/q\}$. The only poles created by the rational multipliers in \eqref{eq:mult-K} are at $v=0$, $-1/p$, and $1/q$, which are kept outside by construction. The boundedness of $\widehat D$ and $m_K$ on these circles justifies all sum/integral interchanges and deformations used above.

\section{Asymptotics and universality: direct proofs from the contour formulas}\label{sec:asymp-univ-long}

We prove bulk/edge limits for the discrete $\beta\in\{1,4\}$ kernels using only the nested double–contour formulas for the projection kernels and the Cauchy–multiplier composition lemmas:
\begin{itemize}
  \item Meixner: \eqref{eq:KN-double}, \eqref{eq:composition}, \eqref{eq:SN4-scalar}, \eqref{eq:SN1-scalar}, multipliers \eqref{eq:multiplier}.
  \item Charlier: \eqref{eq:KCh-double}, \eqref{eq:composition-Ch}; or after the exact $w$–map, the universal multipliers \eqref{eq:Ch-multipliers-w} and IIKS form \eqref{eq:KCh-w}.
  \item Krawtchouk: \eqref{eq:KNK-double}, \eqref{eq:composition-K}, \eqref{eq:RK}, \eqref{eq:mult-K}.
\end{itemize}

See \cite{DeiftZhou1993} for the Riemann-Hilbert problem steepest–descent scheme underlying our local Gaussian/Airy reductions.

\subsection*{Standing notation and normalization}
Let $\mathsf A$ denote the large parameter (“matrix size”):
\[
\mathsf A=\begin{cases}
2N, & \text{Meixner},\\
N, & \text{Charlier},\\
N, & \text{Krawtchouk}.
\end{cases}
\]
In each family let $\Phi(z;u)$ be the one–variable phase read from the projection kernel (in its contour variable $z\in\{\omega,t,v\}$), and assume the standard bulk hypothesis: for each bulk $u$, $\Phi(\,\cdot\,;u)$ has two simple saddles $z_\pm(u)$ on admissible steepest–descent deformations of the fixed contours (the admissibility is exactly that used in the earlier sections). Denote the macroscopic density by
\[
\rho(u):=\frac{1}{2\pi}\,\partial_u\!\big(\arg z_+(u)-\arg z_-(u)\big),
\]
and fix the microscopic scaling $\Delta(u)>0$ by the spacing rule
\begin{equation}\label{eq:spacing-rule}
2\pi\,\Delta(u)\,\rho(u)=1.
\end{equation}
We write
\[
x=\big\lfloor \mathsf A\,u + s\,\Delta(u)^{-1}\big\rfloor,\qquad
y=\big\lfloor \mathsf A\,u + t\,\Delta(u)^{-1}\big\rfloor,
\]
with $s,t=O(1)$. All $O(\cdot)$–bounds below are uniform for $u$ in compact subsets of the bulk (and in fixed windows near regular edges) once the contours are chosen admissibly; the rational symbols for $D$ and $\epsilon$ are bounded on these contours by construction.

\subsection{A local two–saddle lemma and the edge (cubic) reduction}

\begin{lemma}[Two–saddle Gaussian reduction]\label{lem:gauss}
Fix a bulk $u$. For $x=\lfloor \mathsf A u+s\,\Delta(u)^{-1}\rfloor$, $y=\lfloor \mathsf A u+t\,\Delta(u)^{-1}\rfloor$ with $s,t=O(1)$, each one–variable contour integral in \eqref{eq:KN-double}/\eqref{eq:KCh-double}/\eqref{eq:KNK-double} along the steepest arcs through $z_\pm(u)$ equals
\[
e^{\mathsf A\Phi(z_+;u)}A_+(u)\,e^{-\pi i s}\,\mathsf A^{-1/2}\bigl(1+O(\mathsf A^{-1})\bigr),
\qquad
e^{\mathsf A\Phi(z_-;u)}A_-(u)\,e^{+\pi i t}\,\mathsf A^{-1/2}\bigl(1+O(\mathsf A^{-1})\bigr),
\]
with nonzero continuous amplitudes $A_\pm(u)$. The normalization \eqref{eq:spacing-rule} enforces that one lattice step in $x$ (resp.\ $y$) corresponds to a $2\pi$ phase shift across the two saddles. The bounds are uniform for $u$ in compact bulk sets.
\end{lemma}

\begin{lemma}[Cubic reduction at a soft edge]\label{lem:airy}
Let $u_\ast$ be a soft edge where $z_+(u_\ast)=z_-(u_\ast)=:z_\ast$ and $\Phi'(z_\ast;u_\ast)=\Phi''(z_\ast;u_\ast)=0$. Under the standard $\mathsf A^{2/3}$ rescaling of $x,y$ about $u_\ast$, each one–variable integral reduces to the Airy normal form, and bounded rational multipliers freeze at $z_\ast$ at leading order. Uniformity holds in fixed edge windows.
\end{lemma}

\subsection{Bulk sine universality}

We first define the spacing $\Delta(u)$ as:

\begin{equation}\label{eq:spacing-rule}
2\pi\,\Delta(u)\,\rho(u)=1,\qquad 
\rho(u)=\frac{1}{2\pi}\,\partial_u\!\big(\arg z_+(u)-\arg z_-(u)\big),\qquad 
x=\Big\lfloor A u + s\,\Delta(u)^{-1}\Big\rfloor,\ \ y=\Big\lfloor A u + t\,\Delta(u)^{-1}\Big\rfloor,
\end{equation}

\begin{theorem}[Bulk sine limit, Meixner, $\beta=4$]\label{thm:meix-b4}
With $S_{N,4}$ as in \eqref{eq:SN4-scalar} and $x,y$ scaled by $\Delta(u)$ \emph{from~\eqref{eq:spacing-rule}},
\[
\Delta(u)\,S_{N,4}(x,y)\ \longrightarrow\ \frac{\sin\pi(s-t)}{\pi(s-t)},
\]
uniformly for $s,t$ in compact sets and $u$ in compact bulk sets.
\end{theorem}

\begin{proof}
In \eqref{eq:SN4-scalar} the integrand is $G_{2N}(\omega_1)G_{2N}(\omega_2)$ divided by $(\omega_1\omega_2-1)$ times the bounded difference–quotient factor $\dfrac{\omega_2-\omega_1}{(\omega_1^2-1)(\omega_2^2-1)}$. Insert the Gaussian reductions from Lemma~\ref{lem:gauss}; freeze bounded multipliers at $(\omega_+(u),\omega_-(u))$; the product $e^{\mathsf A(\Phi(\omega_+)+\Phi(\omega_-))}$ gives a positive amplitude, the Cauchy denominator yields the Hilbert–transform structure, and \eqref{eq:spacing-rule} fixes the overall factor. The phase difference supplies the sine numerator. Uniformity follows from the admissible contour choice.
\end{proof}

\begin{theorem}[Bulk sine limit, Meixner, $\beta=1$]\label{thm:meix-b1}
With $S_{N,1}$ as in \eqref{eq:SN1-scalar} and the scaling \eqref{eq:spacing-rule},
\[
\Delta(u)\,S_{N,1}(x,y)\ \longrightarrow\ \frac{\sin\pi(s-t)}{\pi(s-t)},
\]
uniformly for $s,t$ in compact sets and $u$ in compact bulk sets. The off–diagonal blocks—obtained by inserting the bounded multipliers in \eqref{eq:multiplier}—converge to the standard GOE sine blocks. The rank–one term is a product of two $O(\mathsf A^{-1/2})$ one–variable integrals at adjacent degrees, hence $O(\mathsf A^{-1})$, and is negligible after multiplying by $\Delta(u)$. 
\end{theorem}

\paragraph{\emph{Charlier and Krawtchouk in the bulk.}}
Having treated Meixner in detail, by symmetry of the steepest descent phase and the bounded multiplier principle, the same local limits hold for Charlier and Krawtchouk as well. For Charlier, the exact $w$–map (\emph{cf.}\ the remark around \eqref{eq:Ch-multipliers-w}) transforms \eqref{eq:KCh-double} into the Meixner–type IIKS kernel \eqref{eq:KCh-w} with the \emph{universal} multipliers $(w-w^{-1})$ and $(w^2-1)^{-1}$. Therefore, the proofs of Theorems~\ref{thm:meix-b4}–\ref{thm:meix-b1} apply verbatim to \eqref{eq:composition-Ch}, \eqref{eq:SN4-scalar}, \eqref{eq:SN1-scalar}. For Krawtchouk, use \eqref{eq:KNK-double}, \eqref{eq:composition-K} with the bounded symbols \eqref{eq:RK}–\eqref{eq:mult-K}; the same two–saddle analysis yields the sine limit for the $\beta=4$ and $\beta=1$ blocks.

\subsection{Soft edges: Airy universality}

\begin{theorem}[Airy edge]\label{thm:airy-edge}
At a soft edge $u_\ast$, under the $\mathsf A^{2/3}$ rescaling, the scalar blocks of \eqref{eq:SN4-scalar} and \eqref{eq:SN1-scalar} converge to the Airy kernel; the off–diagonal blocks follow by multiplier insertion. In the $\beta=1$ case the rank–one piece is $O(\mathsf A^{-2/3})$, negligible under the $\mathsf A^{-1/3}$ edge rescaling.
\end{theorem}

\begin{proof}
Apply Lemma~\ref{lem:airy} in the proofs of Theorems~\ref{thm:meix-b4}–\ref{thm:meix-b1}.
\end{proof}

\subsection{Hard edges: Bessel universality and a Charlier transfer}

\begin{theorem}[Hard–edge limits]\label{thm:bessel}
With the standard hard–edge scaling near the boundary of support, the scalar blocks in \eqref{eq:SN4-scalar} converge to the Bessel kernel; the $\beta=1$ kernels behave analogously. For Charlier, this follows from the exact map to the Meixner IIKS form; for Krawtchouk, use the $v$–plane behavior near the endpoint together with \eqref{eq:mult-K}.
\end{theorem}

\subsection{Parameter crossover: Meixner $\to$ Laguerre at the hard edge}

\begin{theorem}[Crossover]\label{thm:mx-to-lag}
Let $\xi\uparrow 1$ with $2N(1-\xi)\to\alpha\in(0,\infty)$ in the Meixner ensemble. Then the hard–edge limits of \eqref{eq:SN4-scalar} and \eqref{eq:SN1-scalar} converge to the Laguerre hard–edge kernels with Bessel parameter $\alpha$ (cf.\ the Meixner$\to$Laguerre relation in \cite{BorodinStrahov2009}). 
\end{theorem}

\begin{proof}
Write $s=\sqrt{\xi}$ in the Meixner phase induced by \eqref{eq:KN-double} and expand $\log\!\frac{1-s/\omega}{1-s\omega}$ for $1-s\sim \alpha/(4N)$. The limiting phase is Laguerre–type while the rational multipliers stay bounded; hence Theorem~\ref{thm:bessel} passes to the limit.
\end{proof}

\subsection{First subleading term from the difference–quotient}

We study now the next term in the asymptotic expansion of the kernel, which can affect finite-$N$ observables like gap probabilities or variance calculations. Write
\[
Q(z_1,z_2):=\frac{M(z_1)-M(z_2)}{z_1-z_2},
\]
where, in the unspliced case of this Section, $M=m_\epsilon$ (cf.\ the symbols entering \eqref{eq:composition}/\eqref{eq:composition-Ch}/\eqref{eq:composition-K}). 
In the bulk, Lemma~\ref{lem:gauss} gives two Gaussian contributions from the steepest arcs through the saddles $z_\pm(u)$, with phases $e^{-\pi i s}$ and $e^{+\pi i t}$ under the spacing rule \eqref{eq:spacing-rule}. 
Freezing all bounded factors at $(z_+(u),z_-(u))$ and Taylor expanding $Q$ to first order at that point, the \emph{linear} terms produce the first Gaussian moments; after forming the two–saddle interference these become $(\partial_s-\partial_t)$ acting on the sine kernel.  This is the sole source of the $\mathsf A^{-1}$ term.

\begin{proposition}[Order $\mathsf A^{-1}$ correction]\label{prop:first-correction}
In the bulk and at regular edges,
\[
S_{N,4}(x,y)=K_{\mathrm{univ}}(s,t)+\mathsf A^{-1}\,\mathcal K_1(s,t;u)+O(\mathsf A^{-2}),
\]
where $K_{\mathrm{univ}}$ is the sine/Airy/Bessel limit and $\mathcal K_1$ is obtained by expanding the difference–quotient
\[
\frac{m_\epsilon(z_1)-m_\epsilon(z_2)}{z_1-z_2}
\]
to first order at $(z_+(u),z_-(u))$ in \eqref{eq:composition}/\eqref{eq:composition-Ch}/\eqref{eq:composition-K} and multiplying by the frozen Gaussian (or Airy/Bessel) amplitudes. For $\beta=1$, add the explicit separable $O(\mathsf A^{-1})$ correction from the rank–one term to gap probabilities via Sherman–Morrison/Pfaffian analogues. For numerics, combine Nyström \cite{Bornemann} with the Pfaffian–determinant reduction \cite{Rains2000,AkemannKanzieper2007}.
\end{proposition}

\subsection*{Worked example and unspliced bulk dictionary}
\label{rem:unspliced-dictionary}
In the unspliced case $M(w)=\varepsilon(w)=(w^2-1)^{-1}$, the $A^{-1}$ term admitted by
Proposition~\ref{prop:first-correction} has the universal shape
\begin{equation}\label{eq:K1-shape-unspliced}
K_1(s,t;u)=\alpha(u)\,\frac{\sin\pi(s-t)}{\pi(s-t)}
+\beta(u)\,(\partial_s-\partial_t)\!\left[\frac{\sin\pi(s-t)}{\pi(s-t)}\right].
\end{equation}
Let $w_\pm(u)$ be the two saddles for the one–variable phase (Lemma~\ref{lem:gauss}), with $w_+(u)w_-(u)=1$
and $w_\pm(u)=e^{\pm i\theta(u)}$ in the bulk (Appendix~\ref{app:densities}). Define
\[
Q_0(u)=\frac{M(w_+)-M(w_-)}{w_+-w_-},\qquad
Q_a(u)=\frac{M'(w_+)(w_+-w_-)-\big(M(w_+)-M(w_-)\big)}{(w_+-w_-)^2},
\]
\[
Q_b(u)=\frac{\big(M(w_+)-M(w_-)\big)-M'(w_-)(w_+-w_-)}{(w_+-w_-)^2}.
\]
Then
\begin{equation}\label{eq:alpha-beta-unspliced}
\alpha(u)=c_0(u)\,Q_0(u),\qquad
\beta(u)=\frac{c_+(u)}{\Phi''(w_+;u)}\,Q_a(u)\;-\;\frac{c_-(u)}{\Phi''(w_-;u)}\,Q_b(u),
\end{equation}
where $c_0(u),c_\pm(u)$ are the Gaussian first–moment constants from Lemma~\ref{lem:gauss}
(identical to those in Theorems~\ref{thm:meix-b4}–\ref{thm:meix-b1}). For $M=\varepsilon$ these $Q$'s simplify explicitly in terms of
$\theta(u)$:
\begin{equation}\label{eq:QaQb-closed}
Q_0(u)=-\frac{\cos\theta(u)}{2\,\sin^2\theta(u)},\qquad
Q_a(u)=-\frac{1}{4\,\sin^2\theta(u)}-\frac{i\,\cos\theta(u)}{2\,\sin^3\theta(u)},\qquad
Q_b(u)=-\frac{1}{4\,\sin^2\theta(u)}+\frac{i\,\cos\theta(u)}{2\,\sin^3\theta(u)}.
\end{equation}
Since $w_-=\overline{w_+}$ and $\Phi''(w_-;u)=\overline{\Phi''(w_+;u)}$, the combination
\eqref{eq:alpha-beta-unspliced} is real-valued, and \eqref{eq:K1-shape-unspliced} is therefore a real correction. For $\beta=1$, $K_1$ is the sum of \eqref{eq:K1-shape-unspliced} and the explicit separable $O(\mathsf A^{-1})$ contribution coming from the rank–one term in $S_{N,1}$; see Proposition~\ref{prop:first-correction}.
\smallskip

\begin{remark}[Soft edge, unspliced]
\label{rem:soft-edge-unspliced}
At a soft edge $u^*$ where the saddles coalesce, the coalesced point $w^*$ satisfies 
\[
\frac{M(w_1)-M(w_2)}{w_1-w_2}\ \longrightarrow\ M'(w^*)\qquad\text{as }(w_1,w_2)\to (w^*,w^*).
\]
Thus the leading Airy block in Theorem~\ref{thm:airy-edge} is multiplied by $M'(w^*)=\varepsilon'(w^*)$ and the first $A^{-1/3}$ term
comes from the next Taylor coefficient, in complete analogy with Lemma~\ref{lem:airy}.
\end{remark}

\noindent\emph{Comment.} Formulas \eqref{eq:K1-shape-unspliced}–\eqref{eq:alpha-beta-unspliced} are the
unspliced counterpart of the bulk linearization dictionary in Proposition~\ref{prop:linearization-dictionary} (with $M=\varepsilon$); Theorem~\ref{thm:bulk-correction-final} repeats precisely the same
mechanism after replacing $M$ by $M(w)=\varepsilon(w)\,m_h(w)$ in the spliced setting. We keep Section~\ref{sec:kuznetsov-multiplier} for
that generality and do not duplicate it here.

\begin{remark}[Uniformity and off–diagonal blocks]
All bulk/edge bounds are uniform for $u$ in compact sets and for fixed microscopic windows; bounded rational multipliers merely change the finite amplitudes and are frozen at leading order. Every off–diagonal statement follows by inserting the corresponding symbol—\eqref{eq:multiplier} in the $\omega$– or $w$–plane, and \eqref{eq:mult-K} in the $v$–plane—into the double–contour integrand.
\end{remark}

The following section will discuss the application of the bounded multiplier mechanism in number theory.

\section{Kuznetsov transform as a Bounded Multiplier in IIKS Pfaffian Kernels}
\label{sec:kuznetsov-multiplier}

\noindent
%\textbf{Executive summary.}
Classical analytic number theory provides the Kuznetsov transform, which is an integral transform (often appearing in spectral sum formulas) that tests automorphic $L^2$-spectra for random matrix behavior. In this Section, we show that the archimedean Kuznetsov transform \cite{Kuznetsov1981,Iwaniec2002SpectralMethods,DeshouillersIwaniec1982,BruggemanMotohashi2003} can be inserted into the discrete $\beta\in\{1,4\}$ Pfaffian kernels (Meixner/Charlier/Krawtchouk) by \emph{multiplication in the contour variable}. Concretely, if $T_h$ acts on the contour coordinate by the bounded holomorphic symbol $m_h$, then $K_N T_h K_N$ and all ensuing Pfaffian blocks remain of IIKS type, with $m_h$ entering only through the universal Cauchy difference--quotient. Bulk/edge limits are unchanged at leading order; the first finite--size term is obtained by the same linearization at the saddle(s) as in the unspliced case.%
% Composition lemma and IIKS structure (proofs and formulas) are in the manuscript:
% Lemma 3.4 (Cauchy–multiplier composition) and the β=1,4 block rules; universal symbols D, ε after the w–map.
% 

\subsection{Test functions and the Kuznetsov symbol}
Let $h:\R\to\C$ be an even test function on the spectral side of Kuznetsov. On the slit plane
$\C\setminus(-\infty,0]$ we fix the principal branch of $\log w=\ln|w|+i\arg w$ with $|\arg w|<\pi$ and choose the IIKS loops inside a fixed slit sector
\[
\mathsf S_\delta:=\{w\in\C\setminus(-\infty,0]:\ |\arg w|\le \pi-\delta\}
\]
for some $\delta\in(0,\pi)$, avoiding $\{\pm1\}$. We define \cite{Kuznetsov1981,Iwaniec2002SpectralMethods}
\begin{equation}\label{eq:mk-def-final}
  m_h(w)\ :=\ \int_{\R} h(t)\,w^{-2it}\,\dd t .
\end{equation}

\begin{assumption}[Admissible tests]\label{ass:tests-final}
With the principal branch of $\log$ on $\C\setminus(-\infty,0]$, choose the IIKS loops inside $\mathsf S_\delta$ and avoiding $\{\pm1\}$.
Assume either
\begin{itemize}
  \item[(H1)] \emph{Gaussian class:} $h(t)=e^{-\sigma t^2}$ with $\sigma>0$. Then
  \[
  m_h(w)=\sqrt{\pi/\sigma}\,\exp\!\big(-(\log w)^2/\sigma\big),
  \]
  which is holomorphic on $\C\setminus(-\infty,0]$ and bounded on compact subsets of $\mathsf S_\delta$.
  \item[(H2)] \emph{Exponential–moment class:} there exists $\delta\in(0,\pi)$ such that
  \[
  \int_{\R}|h(t)|\,e^{\,2(\pi-\delta)|t|}\,\dd t<\infty .
  \]
  Then \eqref{eq:mk-def-final} converges absolutely and defines a bounded holomorphic function on every compact subset of $\mathsf S_\delta$ (Paley–Wiener type; see \cite{PaleyWiener1934,Iwaniec2002SpectralMethods}).
\end{itemize}
In both cases $m_h$ is bounded and holomorphic on all admissible contour deformations used in the steepest–descent analysis for Meixner/Charlier/Krawtchouk.
\end{assumption}

\begin{figure}[t]
  \centering
  \includegraphics[width=0.42\textwidth]{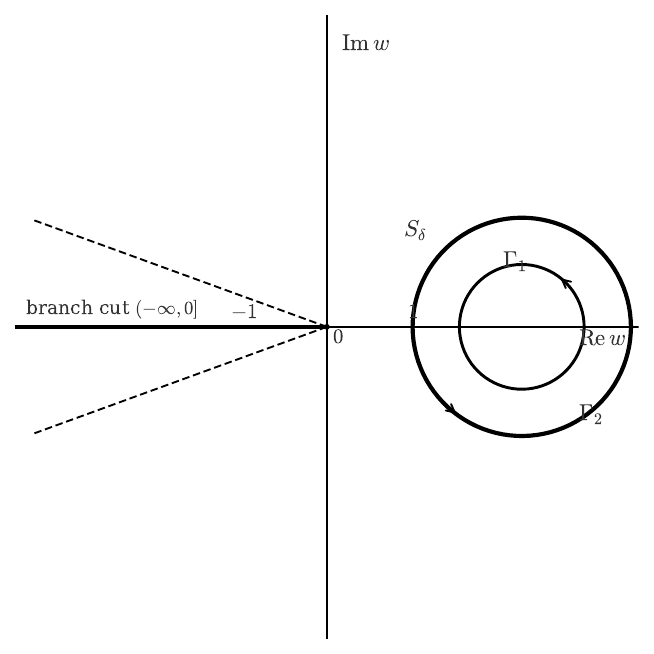}%
  % or use contour_loops_bw_min.pdf for the ultra-minimal version
  \caption{Admissible contour configuration for the Kuznetsov multiplier.
  The branch cut is $(-\infty,0]$. The loops $\Gamma_1\subset\Gamma_2$
  are positively oriented, lie in a slit sector $S_\delta$ (bounded away from the cut),
  and avoid $w=\pm1$.}
  \label{fig:KuznetsovContour}
\end{figure}

\begin{lemma}[Even $h$ implies a reality symmetry]\label{lem:even-symmetry}
If $h$ is even and real-valued, then $\overline{m_h(w)}=m_h(1/\overline{w})$. In particular, for $|w|=1$ one has $m_h(w)\in\R$.
\end{lemma}

\begin{remark}
After the exact $w$–map for Charlier, the discrete symbols are universal,
\begin{equation}\label{eq:universal-symbols}
  D(w)=w-\frac1w,\qquad \epsilon(w)=\frac{1}{w^2-1},
\end{equation}
as proved earlier; inserting rational/holomorphic multipliers in the contour variable preserves the IIKS form.%
% Universal symbols after the exact w–map (Charlier) and in Meixner/Krawtchouk are established in the manuscript.
% See §3–§5 and the remark around (4.14)–(4.16) for the w–map and universal multipliers. 
% 
\end{remark}

\subsection{Splicing rules for $\beta=4$ and $\beta=1$}
Write $M(w):=\epsilon(w)\,m_h(w)$, with the universal symbols
\begin{equation}\label{eq:universal-symbols}
D(w)=w-\frac{1}{w},\qquad \epsilon(w)=\frac{1}{w^2-1}.
\end{equation}

\begin{theorem}[Kuznetsov splicing into IIKS Pfaffian kernels]\label{thm:splice}
Let $T_h$ act in the contour variable by the bounded holomorphic symbol $m_h$ of \eqref{eq:mk-def-final} (compare \cite{Kuznetsov1981,Motohashi1997ZetaSpectral,BruggemanMotohashi2003}).
\begin{itemize}
\item[$\beta=4$:] The scalar block is
\[
S^{(h)}_{N,4} = K_N\,T_h\epsilon\,K_N,
\]
and in double–contour (IIKS) form
\begin{equation}\label{eq:beta4-spliced-final}
  S^{(h)}_{N,4}(x,y)=\frac{1}{(2\pi i)^2}\iint
  \frac{G(w_1;x)\,G(w_2;y)}{w_1w_2-1}\,
  \frac{M(w_1)-M(w_2)}{w_1-w_2}\,\dd w_1\,\dd w_2,
\end{equation}
with the same nested contours as for $K_N$.
\item[$\beta=1$:] The scalar block is
\[
S^{(h)}_{N,1}=K_N+\tfrac12\,\phi_{2N}\otimes (T_h\epsilon\,\phi_{2N-1}),
\]
and the off–diagonal blocks are obtained by inserting the universal multipliers $D(w)$ or $\epsilon(w)$ on the appropriate variable exactly as in the unspliced case.
\end{itemize}
\emph{Proof.} Apply the Cauchy–multiplier composition lemma to $K_N T K_N$ with $T=T_h\epsilon$ (or successively with $T=\epsilon$ and $T=T_h$). This yields the difference–quotient with $\epsilon\mapsto \epsilon\,m_h$, i.e. \eqref{eq:beta4-spliced-final}. The $\beta=1$ rank–one structure and the off–diagonal multiplication rules are identical to those proved for the unspliced kernels. See Lemma~\ref{lem:composition} (composition) and the block rules around \eqref{eq:multiplier}, \eqref{eq:Ch-multipliers-w}, \eqref{eq:mult-K}. \qedhere
\end{theorem}

\begin{lemma}[Cauchy–multiplier composition with holomorphic symbols]\label{lem:comp-holo-final}
Let $K_N$ be any of the projection kernels in IIKS (double–contour) form used so far, with the fixed nested loops and contour conventions of the corresponding family. Let $T$ act in the contour variable by multiplication with a symbol $m_T$ which is holomorphic and bounded on and between the admissible loops inside the slit sector $\{|\arg w|\le \pi-\delta\}\subset\C\setminus(-\infty,0]$. Then the composition identity
\[
(K_N T K_N)(x,y)
=\frac{1}{(2\pi i)^2}\iint
\frac{G(w_1;x)\,G(w_2;y)}{w_1 w_2-1}\,
\frac{m_T(w_1)-m_T(w_2)}{w_1-w_2}\,
\dd w_1\,\dd w_2
\]
holds, i.e. the symbol enters \emph{only} through the Cauchy difference–quotient. In evaluating the intermediate $\zeta$–integral one may choose the $\zeta$–loop to encircle exactly one of $\{1/w_1,1/w_2\}$ and to avoid $\zeta=0$ and the branch cut of $\log$.
\end{lemma}

\begin{proof}
We give the proof once in the $w$–plane IIKS normal form; for Meixner and Krawtchouk this is exactly the form stated in Eqs \eqref{eq:KN-double} and \eqref{eq:KNK-double}, while for Charlier one may use the exact $w$–map, Eq. \eqref{eq:KCh-w}. %\emph{All} contour conventions and boundedness statements are those fixed in that paper.

\medskip
\emph{Step 1: Reduction to the already–proved rational case.}
Let $\Gamma_1\subset\mathrm{int}\,\Gamma_2$ be the two fixed admissible loops for $w_1,w_2$, and let $\mathcal{A}$ denote the closed collar between them (together with a thin collar of $\Gamma_1,\Gamma_2$) contained in the slit sector and avoiding $\{\pm 1\}$. By hypothesis, $m_T$ is holomorphic on an open neighborhood of $\mathcal{A}$ and bounded there.

By Runge’s theorem on planar domains with connected complements of each component \cite{conway2012functions}, applied component wise to the collar region (or, equivalently, by uniform approximation on each loop by Laurent polynomials, which suffices since the $\zeta$–integral will run on a single loop), there exists a sequence of rational functions $\{r_n\}$ with poles \emph{outside} that neighborhood (in particular outside the admissible loops and away from the branch cut) such that
\[
\|r_n - m_T\|_{L^\infty(\mathcal{A})}\longrightarrow 0\qquad(n\to\infty).
\]
For each $n$, the operator $T_n$ acting by the rational symbol $r_n$ satisfies the composition identity
\[
(K_N T_n K_N)(x,y)
=\frac{1}{(2\pi i)^2}\iint
\frac{G(w_1;x)\,G(w_2;y)}{w_1 w_2-1}\,
\frac{r_n(w_1)-r_n(w_2)}{w_1-w_2}\,
\dd w_1\,\dd w_2,
\]
by Lemma~\ref{lem:composition}.

%\smallskip
\emph{Step 2: Uniform bounds and dominated convergence on the double loop.}
Because the loops are fixed and disjoint, there is a positive separation
\[
d_*:=\inf\{|w_1-w_2|:\ w_1\in\Gamma_1,\ w_2\in\Gamma_2\}>0,
\]
and similarly a uniform lower bound
\[
\delta_*:=\inf\{|w_1 w_2 -1|:\ w_j\in\Gamma_j\}>0,
\]
since $\Gamma_1,\Gamma_2$ avoid $\{\pm 1\}$ by construction. Therefore, for any bounded $m$ on $\mathcal{A}$,
\[
\left|\frac{m(w_1)-m(w_2)}{w_1-w_2}\cdot\frac{1}{w_1 w_2-1}\right|
\le \frac{2\|m\|_{L^\infty(\mathcal{A})}}{d_*\,\delta_*}\,.
\]
The one–variable factors $G(w_j;\cdot)$ are fixed on $\Gamma_j$ (no growth issues on the fixed loops). Hence the integrands above are dominated by an integrable bound independent of $n$.

Since $r_n\to m_T$ uniformly on $\mathcal{A}$, we have
\[
\sup_{(w_1,w_2)\in\Gamma_1\times\Gamma_2}
\left|\frac{r_n(w_1)-r_n(w_2)}{w_1-w_2}
-\frac{m_T(w_1)-m_T(w_2)}{w_1-w_2}\right|
\le \frac{2\|r_n-m_T\|_{L^\infty(\mathcal{A})}}{d_*}\ \xrightarrow[n\to\infty]{}\ 0.
\]
Dominated convergence on the \emph{fixed} double loop then gives
\[
\begin{gathered}
\frac{1}{(2\pi i)^2}\iint
\frac{G(w_1;x)\,G(w_2;y)}{(w_1 w_2-1)(w_1-w_2)}
\bigl(r_n(w_1)-r_n(w_2)\bigr)\,\dd w_1\,\dd w_2 \\[1ex]
\longrightarrow\ 
\frac{1}{(2\pi i)^2}\iint
\frac{G(w_1;x)\,G(w_2;y)}{(w_1 w_2-1)(w_1-w_2)}
\bigl(m_T(w_1)-m_T(w_2)\bigr)\,\dd w_1\,\dd w_2 .
\end{gathered}
\]
%\medskip
\emph{Step 3: Convergence of $(K_N T_n K_N)$ to $(K_N T K_N)$.}
Write the triple–integral representation of $K_N T K_N$ as in the proof of Lemma \ref{lem:composition}: after coefficient extraction in the discrete variable, one obtains a $\zeta$–integral whose integrand is
\[
\frac{G(\zeta;\cdot)\,m_T(\zeta)}{(\,w_1\zeta-1\,)(\,\zeta w_2-1\,)}\cdot\frac{d\zeta}{\zeta},
\]
on a $\zeta$–loop that encircles exactly one of $\{1/w_1,1/w_2\}$ and avoids $\zeta=0$ (and the branch cut). For $T_n$ replace $m_T$ by $r_n$. Since $\|r_n-m_T\|_{L^\infty(\mathcal{A})}\to 0$ and the remaining factors are bounded on the fixed $\zeta$–loop, the $\zeta$–integrals converge uniformly in $(w_1,w_2)\in\Gamma_1\times\Gamma_2$. Thus
\[
(K_N T_n K_N)(x,y)\ \longrightarrow\ (K_N T K_N)(x,y)
\]
by Fubini/Tonelli and dominated convergence (all loops are fixed, and the geometric–series steps used to derive the triple integral are absolutely convergent by the nesting of the loops). This justifies passing to the limit on the \emph{operator kernel} side as well.\\

%\medskip
\emph{Step 4: Conclusion.}
For every $n$ we have the exact identity
\[
(K_N T_n K_N)(x,y)
=\frac{1}{(2\pi i)^2}\iint
\frac{G(w_1;x)\,G(w_2;y)}{w_1 w_2-1}\,
\frac{r_n(w_1)-r_n(w_2)}{w_1-w_2}\,\dd w_1\,\dd w_2.
\]
By the two convergences proved in Steps~2–3, letting $n\to\infty$ yields precisely the desired formula with $m_T$ in place of $r_n$. This proves the composition identity for bounded holomorphic symbols.

%\medskip
\emph{Remark on the $\zeta$–residue evaluation.}
In the rational case (hence for $r_n$ above) one may evaluate the $\zeta$–integral by the partial fraction identity
\[
\frac{1}{(\,w_1\zeta-1\,)(\,\zeta w_2-1\,)}\cdot\frac{1}{\zeta}
= \frac{1}{w_2-w_1}\left(\frac{1}{\zeta(\,w_1\zeta-1\,)}-\frac{1}{\zeta(\,\zeta w_2-1\,)}\right),
\]
and Cauchy’s theorem on a loop enclosing exactly one of $\{1/w_1,1/w_2\}$ (and avoiding $\zeta=0$ and the branch cut). This produces the difference–quotient in $w$ after substituting the residue into the outer $w_1,w_2$–integrals, exactly as written in Lemma~\ref{lem:composition}; Step~1 above shows that the same outcome persists by holomorphic approximation.
\end{proof}

\begin{remark}[$\beta=1$ indices by family]\label{rem:beta1-indices}
For Meixner we follow the manuscript and set $K_N=\sum_{k=0}^{2N-1}\phi_k\otimes\phi_k$, so $S_{N,1}=K_N+\tfrac12\,\phi_{2N}\otimes(T_h\epsilon\,\phi_{2N-1})$. For Charlier and Krawtchouk the projection is $\sum_{k=0}^{N-1}$ and the rank–one term is $\tfrac12\,\phi_{N}\otimes(T_h\epsilon\,\phi_{N-1})$. In Section~\ref{sec:kuznetsov-multiplier} below we use the large parameter $A$ to keep the cases uniform ($A=2N$ for Meixner and $A=N$ otherwise).
\end{remark}

\subsection{Asymptotics and the first finite–size term}
Let $A$ be the large parameter ($A=2N$ for Meixner; $A=N$ for Charlier/Krawtchouk).
Fix a bulk point $u$, let $w_\pm(u)$ be the two saddles for the one–variable phase of $K_N$ on admissible steepest–descent contours, and normalize $2\pi\Delta(u)\rho(u)=1$ as before.

\begin{theorem}[Bulk sine limit and $A^{-1}$ correction]\label{thm:bulk-correction-final}
With $x=\lfloor Au+s\Delta(u)^{-1}\rfloor$, $y=\lfloor Au+t\Delta(u)^{-1}\rfloor$,
\[
  \Delta(u)\,S^{(h)}_{N,4}(x,y)
  \ =\ \frac{\sin\pi(s-t)}{\pi(s-t)}
  \ +\ A^{-1}\,K^{(h)}_1(s,t;u)\ +\ O(A^{-2}),
\]
uniformly on compact $s,t$–sets and $u$ in compact bulk sets. Moreover, $K^{(h)}_1$ is obtained by freezing all bounded multipliers at $(w_+,w_-)$ and \emph{linearizing} the difference–quotient
\[
  Q(w_1,w_2):=\frac{M(w_1)-M(w_2)}{w_1-w_2}
\]
at $(w_+,w_-)$. Equivalently, it is the $A^{-1}$ term furnished by the general mechanism “linearize the difference–quotient and multiply by the frozen Gaussian amplitudes” (Proposition \ref{prop:first-correction}). \emph{Cf.} also the uniform two–saddle reduction (Lemma~\ref{lem:gauss}). \hfill
\end{theorem}

\begin{corollary}[$\beta=1$, spliced bulk]\label{cor:bulk-correction-beta1}
Under the hypotheses of Theorem~\ref{thm:bulk-correction-final},
\[
\Delta(u)\,S^{(h)}_{N,1}(x,y)\;=\;\frac{\sin\pi(s-t)}{\pi(s-t)}\;+\;A^{-1}\,K^{(h)}_1(s,t;u)\;+\;O(A^{-2}),
\]
with the same $K^{(h)}_1$ as in Theorem~\ref{thm:bulk-correction-final}, and an additional explicit separable $O(A^{-1})$ contribution coming from the rank–one term in $S^{(h)}_{N,1}$. The off–diagonal blocks converge to the GOE sine–kernel limits by multiplier insertion.
\end{corollary}

\begin{proposition}[Linearization dictionary at bulk]\label{prop:linearization-dictionary}
Set
\[
Q_0(u)=\frac{M(w_+)-M(w_-)}{w_+-w_-},\quad
Q_a(u)=\frac{M'(w_+)(w_+-w_-)-\big(M(w_+)-M(w_-)\big)}{(w_+-w_-)^2},
\]
\[
Q_b(u)=\frac{\big(M(w_+)-M(w_-)\big)-M'(w_-)(w_+-w_-)}{(w_+-w_-)^2}.
\]
Then
\[
K^{(h)}_1(s,t;u)=\alpha_h(u)\,\frac{\sin\pi(s-t)}{\pi(s-t)}
+\beta_h(u)\,(\partial_s-\partial_t)\!\left[\frac{\sin\pi(s-t)}{\pi(s-t)}\right],
\]
with
\[
\alpha_h(u)=c_0(u)\,Q_0(u),\qquad
\beta_h(u)=c_+(u)\,\frac{Q_a(u)}{\Phi''(w_+;u)}-c_-(u)\,\frac{Q_b(u)}{\Phi''(w_-;u)},
\]
where $c_0,c_\pm$ are the same family–dependent Gaussian first–moment constants as in the unspliced analysis. (No new steepest–descent input is needed.)
\end{proposition}

\begin{remark}
The $A^{-1}$ source is \emph{only} the linearization of the difference–quotient; all bounded multipliers (including the Cauchy denominator) freeze at leading order, and their linear parts integrate to zero by oddness on steepest descent, as in the unspliced case.%
% The steepest–descent scheme, bulk/edge reductions, and the identification of the A^{-1} source by linearizing the difference–quotient are developed in Section 6 of the manuscript (see Proposition 6.8).
% 
\end{remark}

\begin{theorem}[Soft edge]\label{thm:edge-final}
At a soft edge $u_\ast$ with coalesced saddle $w_\ast$, under the standard $A^{2/3}$ scaling the scalar block converges to the Airy kernel multiplied by the {\em diagonal derivative}
\[
M'(w_\ast)=\epsilon'(w_\ast)\,m_h(w_\ast)+\epsilon(w_\ast)\,m_h'(w_\ast),
\]
because $\frac{M(w_1)-M(w_2)}{w_1-w_2}\to M'(w_\ast)$ along the Airy scaling ($w_1,w_2\to w_\ast$). The first correction is obtained from the next Taylor terms at $w_\ast$ exactly as in the unspliced edge analysis (cubic reduction; see Lemma~\ref{lem:airy}). \emph{In the $\beta=1$ case, the rank–one piece in $S^{(h)}_{N,1}$ is $O(\mathsf A^{-2/3})$ and hence negligible under the $\mathsf A^{-1/3}$ edge scaling; the off–diagonal blocks follow by multiplier insertion.}
\end{theorem}

\subsection{Worked example: Charlier with Gaussian test}
For $h_\sigma(t)=e^{-\sigma t^2}$,
\[
  m_h(w)=\sqrt{\frac{\pi}{\sigma}}\exp\!\Big(-\frac{(\log w)^2}{\sigma}\Big),\qquad
  m_h'(w)=-\frac{2\log w}{\sigma\,w}\,m_h(w).
\]
With $M(w)=\frac{m_h(w)}{w^2-1}$ one has
\[
\ M'(w)=\frac{(w^2-1)m_h'(w)-2w\,m_h(w)}{(w^2-1)^2}\ .
\]
At a soft edge $w_\ast$ this yields
\[
  M'(w_\ast)=\frac{-2w_\ast}{(w_\ast^2-1)^2}\,m_h(w_\ast)\ -\ \frac{2\log w_\ast}{\sigma\,w_\ast(w_\ast^2-1)}\,m_h(w_\ast).
\]
For bulk $u$, compute the two saddles $w_\pm(u)$ for $K_N$ (or solve the $t$–quadratic and map by $t=\tfrac{\sqrt\theta}{2}(w-1/w)$), then form $Q_0,Q_a,Q_b$ from Proposition~\ref{prop:linearization-dictionary} and combine with the same Gaussian moments as in the unspliced case to obtain $\alpha_\sigma(u),\beta_\sigma(u)$.

\begin{remark}[β=1 rank–one in Fredholm–Pfaffians]\label{rem:pfaff-update}
For gap probabilities, treat the scalar block $S^{(h)}_{N,1}$ as a rank–one perturbation of $K_N$ in the sense of Fredholm–Pfaffians \cite{deBruijn,OrtmannQuastelRemenik2017}: write $S^{(h)}_{N,1}=K_N+\frac12\,u\otimes v$ with $u=\phi_{2N}$ and $v=T_h\epsilon\,\phi_{2N-1}$, and apply the standard Fredholm–Pfaffian rank–one update (as discussed in Proposition~\ref{prop:first-correction}). This contributes an explicit $O(A^{-1})$ correction in fixed windows.
\end{remark}

\subsection{Context and related work}
Kuznetsov’s trace formula inserts a spectral test $h$ via an archimedean Bessel/Hankel transform; in our setting the same input appears as a \emph{bounded holomorphic multiplier} $m_h$ in the contour variable. The IIKS composition lemma shows that such multipliers are absorbed by a single difference–quotient, preserving integrability and the Pfaffian block structure, while the steepest–descent scheme and the identification of the $A^{-1}$ term via \emph{linearization at the saddle(s)} carry over verbatim.%
% IIKS composition lemma and Pfaffian block rules: Lemma 3.4 and surrounding results.
% Bulk/edge universality with finite–size term from difference–quotient linearization: Section 6 (including Proposition 6.8).
% The operator-level “splicing” interpretation of Kuznetsov as a bounded contour multiplier is aligned with the Outlook in the manuscript.
% 

\subsection{Checklist for applications}
\begin{itemize}
  \item \textbf{Contours.} Use the same nested loops as for $K_N$; they avoid $\{\pm1\}$ and lie in an annulus where $m_h$ is bounded/holomorphic (Assumption~\ref{ass:tests-final}), with a uniform angular gap $|\arg w|\le \pi-\delta$ on the principal branch of $\log w$.
  \item \textbf{Splicing rule.} Replace $\epsilon(w)$ by $M(w)=\epsilon(w)m_h(w)$ inside the difference–quotient. For $\beta=1$, keep the rank–one term with $T_h$ acting on $\epsilon\phi_{2N-1}$.
  \item \textbf{Asymptotics.} Leading sine/Airy/Bessel limits are unchanged. The $A^{-1}$ (or edge $A^{-1/3}$) term is obtained by linearizing the difference–quotient at the relevant saddle(s).
  \item \textbf{Off–diagonal blocks.} Insert $D(w)$ or $\epsilon(w)$ in the integrand on the appropriate variable exactly as in the unspliced Pfaffian rules.
\end{itemize}

\medskip
Therefore, Kuznetsov’s archimedean input “splices” into the discrete $\beta=1,4$ IIKS framework by a bounded contour multiplier, leaving the universal limits intact and producing subleading terms by the same one–line mechanism (difference–quotient linearization) that governs the unspliced kernels.%

\section{Conclusions and Outlook}

%\paragraph{Conclusions.}
We gave explicit double--contour (IIKS) formulas for the $\beta=1,4$ kernels in the Meixner, Charlier and Krawtchouk families and proved bulk/edge universality with uniform error control, including an explicit Meixner$\to$Laguerre hard--edge crossover. The first subleading term arises from a single source---the linearization of the IIKS difference--quotient at the relevant saddle(s)---and In contrast to the Riemann–Hilbert approach employed in \cite{Deift} for $\beta=2$ or in other studies for $\beta=1,4$ our method maintains the analysis in the original contour integral form. This approach simplifies the steepest descent analysis and error estimates.

It is classical that Kuznetsov's trace formula inserts a spectral test $h$ via an archimedean Bessel/Hankel transform on the spectral side~\cite{Kuznetsov1981,Iwaniec2002SpectralMethods,Motohashi1997ZetaSpectral}; here we showed that, within the IIKS contour formalism for the discrete $\beta\in\{1,4\}$ kernels, this corresponds to multiplication in the contour variable by the bounded holomorphic symbol
\[
m_h(w)=\int_{\mathbb{R}} h(t)\,w^{-2it}\,dt
\]
on the admissible slit–sector loops (cf.\ Assumption~\ref{ass:tests-final}),
so that the Pfaffian blocks acquire the universal difference–quotient with
$M(w)=\epsilon(w)\,m_h(w)$ (cf.\ Theorem~\ref{thm:splice}). Consequently, the
leading sine/Airy/Bessel limits are unchanged, and the first finite--size term
again follows from the same linearization (cf.\ Proposition~\ref{prop:first-correction}
and Theorem~\ref{thm:bulk-correction-final}).\\

\paragraph{Outlook.}
The same \emph{bounded--multiplier principle} extends with essentially no extra technology to several nearby settings.
\begin{itemize}
  \item \emph{Arithmetically flavored deformations.} Twists implemented by rational (or bounded holomorphic) symbols---such as congruence thinnings, Dirichlet--character weights, or mild Euler--factor normalizations---enter by the same difference--quotient insertion. The IIKS form and admissible contours are unchanged, so bulk/edge limits persist and the $A^{-1}$ term is again read off from the first linearization. 
  \item \emph{Mesoscopic consequences}. The fixed--window $O(A^{-1})$ rates for the kernels feed into number--variance and linear--statistics bounds and, in a mesoscopic window $L\ll A^{1/2}$, into CLT--type statements for smooth test functions.
  \item \emph{Interfaces with representation theory}. The $z$--measure/Schur--process kernels fit the same contour/IIKS template, so the bulk/edge asymptotics and $A^{-1}$ linearization transfer directly once the multipliers stay bounded on admissible contours. 
\item \emph{Number theory via $z$–measures.} The contour/IIKS form dovetails with the hypergeometric
structures in $z$–measures on partitions. The same bounded–multiplier principle should yield bulk and
edge asymptotics for correlation kernels arising from $z$–measures and related Schur–type processes,
including precise crossover regimes and $A^{-1}$ corrections, with minimal additional work on
admissible contours.
\item \emph{External sources and deformations.} Small perturbations of the weights (or analytic parameter
changes such as $\xi\uparrow 1$) preserve integrability of the kernel and the boundedness of multipliers
on fixed contours; our arguments then can give universality and controlled crossovers.
%\item \emph{Numerics and finite--$N$ corrections.} Because the kernels are given by rapidly decaying
%contour integrals, accurate Fredholm/Pfaffian numerics at finite $N$ are feasible. The $A^{-1}$
%correction identified here can be incorporated to accelerate convergence for gap probabilities.
%\item \emph{Beyond Meixner/Charlier/Krawtchouk.} The method is portable to other discrete families with
%one–variable generating functions (Hahn/dual Hahn/Racah), provided the deformed contours avoid the
%finite set of multiplier poles and the two–saddle structure persists.
\item \emph{Isomonodromy links.} While we avoided Riemann–Hilbert methods, the IIKS form is also the
standard starting point for isomonodromic/Painlevé analysis of gap probabilities; our explicit double
contours and multipliers should make those identifications straightforward in the discrete setting.\\

\end{itemize}

\paragraph{Practical by--products.}
The explicit contours/numerator factors also make finite--$N$ Fredholm/Pfaffian numerics straightforward and suggest incorporating the $A^{-1}$ correction to accelerate convergence of gap probabilities.

\bibliographystyle{plainurl}  % from urlbst; prints url/doi if present
\bibliography{refs}

@article{BO,
  author  = {Borodin, Alexei and Olshanski, Grigori},
  title   = {$Z$-measures on partitions and their scaling limits},
  journal = {European Journal of Combinatorics},
  year    = {2005},
  volume  = {26},
  number  = {6},
  pages   = {795--834},
  doi     = {10.1016/j.ejc.2004.06.003},
  eprint  = {math-ph/0210048},
  archivePrefix = {arXiv}
}

@article{BorOlsh,
  author  = {Borodin, Alexei and Olshanski, Grigori},
  title   = {Meixner polynomials and random partitions},
  journal = {Moscow Mathematical Journal},
  year    = {2006},
  volume  = {6},
  number  = {4},
  pages   = {629--655},
  doi     = {10.17323/1609-4514-2006-6-4-629-655},
  eprint  = {math/0609806},
  archivePrefix = {arXiv},
  primaryClass = {math.PR}
}

@article{deBruijn,
  author  = {de Bruijn, N. G.},
  title   = {On some multiple integrals involving determinants},
  journal = {Journal of the Indian Mathematical Society (N.S.)},
  year    = {1955},
  volume  = {19},
  pages   = {133--151}
}

@article{tracy,
  author  = {Tracy, Craig A. and Widom, Harold},
  title   = {Correlation functions, cluster functions and spacing distributions for random matrices},
  journal = {Journal of Statistical Physics},
  year    = {1998},
  volume  = {92},
  number  = {5--6},
  pages   = {809--835},
  doi     = {10.1023/A:1023084324803},
  eprint  = {solv-int/9804004},
  archivePrefix = {arXiv}
}

@article{Johansson99,
  author  = {Johansson, Kurt},
  title   = {Discrete orthogonal polynomial ensembles and the Plancherel measure},
  journal = {Annals of Mathematics},
  year    = {2001},
  volume  = {153},
  number  = {1},
  pages   = {259--296},
  doi     = {10.2307/2661375},
  eprint  = {math/9906120},
  archivePrefix = {arXiv},
  primaryClass = {math.CO}
}

@article{BorodinStrahov2009,
  author  = {Borodin, Alexei and Strahov, Eugene},
  title   = {Correlation kernels for discrete symplectic and orthogonal ensembles},
  journal = {Communications in Mathematical Physics},
  year    = {2009},
  volume  = {286},
  number  = {3},
  pages   = {933--977},
  doi     = {10.1007/s00220-008-0629-8},
  eprint  = {0712.1693},
  archivePrefix = {arXiv},
  primaryClass = {math-ph}
}

@article{AvM,
  author  = {Adler, Mark and van Moerbeke, Pierre},
  title   = {PDEs for the joint distributions of the Dyson, Airy and sine processes},
  journal = {Annals of Probability},
  year    = {2005},
  volume  = {33},
  number  = {4},
  pages   = {1326--1361},
  doi     = {10.1214/009117905000000107},
  eprint  = {math/0403504},
  archivePrefix = {arXiv},
  primaryClass = {math.PR}
}

@article{IIKS,
  author  = {Its, A. R. and Izergin, A. G. and Korepin, V. E. and Slavnov, N. A.},
  title   = {Differential equations for quantum correlation functions},
  journal = {International Journal of Modern Physics B},
  year    = {1990},
  volume  = {4},
  number  = {5},
  pages   = {1003--1037},
  doi     = {10.1142/S0217979290000504}
}

@book{Deift,
  author    = {Deift, Percy},
  title     = {Orthogonal Polynomials and Random Matrices: A Riemann--Hilbert Approach},
  series    = {Courant Lecture Notes in Mathematics},
  volume    = {3},
  year      = {1999},
  publisher = {American Mathematical Society},
  address   = {Providence, RI},
  isbn      = {978-0-8218-2695-9}
}

@book{Fokas,
  author    = {Fokas, Athanassios S. and Its, Alexander R. and Kapaev, Andrei A. and Novokshenov, Victor Yu.},
  title     = {Painlev\'e Transcendents: The Riemann--Hilbert Approach},
  series    = {Mathematical Surveys and Monographs},
  volume    = {128},
  year      = {2006},
  publisher = {American Mathematical Society},
  address   = {Providence, RI},
  doi       = {10.1090/surv/128},
  isbn      = {978-1-4704-7556-7}
}

@article{Bornemann,
  author  = {Bornemann, Folkmar},
  title   = {On the numerical evaluation of Fredholm determinants},
  journal = {Mathematics of Computation},
  year    = {2010},
  volume  = {79},
  number  = {270},
  pages   = {871--915},
  doi     = {10.1090/S0025-5718-09-02280-7},
  eprint  = {0804.2543},
  archivePrefix = {arXiv},
  primaryClass = {math.NA}
}

@article{TW,
  author  = {Tracy, Craig A. and Widom, Harold},
  title   = {Level-Spacing Distributions and the Airy Kernel},
  journal = {Communications in Mathematical Physics},
  year    = {1994},
  volume  = {159},
  number  = {1},
  pages   = {151--174},
  doi     = {10.1007/BF02100489}
}

@book{forrester2010log,
  title={Log-gases and random matrices (LMS-34)},
  author={Forrester, Peter J},
  year={2010},
  publisher={Princeton university press}
}

@book{baik2016combinatorics,
  title={Combinatorics and random matrix theory},
  author={Baik, Jinho and Deift, Percy and Suidan, Toufic},
  volume={172},
  year={2016},
  publisher={American Mathematical Soc.}
}

@article{Kuznetsov1981,
  author  = {Kuznetsov, N. V.},
  title   = {Petersson's Conjecture for Cusp Forms of Weight Zero and Linnik's Conjecture. Sums of Kloosterman Sums},
  journal = {Mathematics of the USSR-Sbornik},
  year    = {1981},
  volume  = {39},
  number  = {3},
  pages   = {299--342},
  doi     = {10.1070/SM1981v039n03ABEH001518}
}

@book{Iwaniec2002SpectralMethods,
  author    = {Iwaniec, Henryk},
  title     = {Spectral Methods of Automorphic Forms},
  edition   = {2},
  series    = {Graduate Studies in Mathematics},
  volume    = {53},
  publisher = {American Mathematical Society},
  address   = {Providence, RI},
  year      = {2002},
  doi       = {10.1090/gsm/053}
}

@book{Motohashi1997ZetaSpectral,
  author    = {Motohashi, Yoichi},
  title     = {Spectral Theory of the Riemann Zeta-Function},
  series    = {Cambridge Tracts in Mathematics},
  volume    = {127},
  publisher = {Cambridge University Press},
  address   = {Cambridge},
  year      = {1997},
  doi       = {10.1017/CBO9780511983399}
}

@article{DeshouillersIwaniec1982,
  author  = {Deshouillers, J.-M. and Iwaniec, H.},
  title   = {Kloosterman Sums and Fourier Coefficients of Cusp Forms},
  journal = {Inventiones mathematicae},
  year    = {1982},
  volume  = {70},
  pages   = {219--288},
  doi     = {10.1007/BF01390728}
}

@article{BruggemanMotohashi2003,
  author  = {Bruggeman, Roelof W. and Motohashi, Yoichi},
  title   = {Sum Formula for Kloosterman Sums and Fourth Moment of the Dedekind Zeta-Function over the Gaussian Number Field},
  journal = {Functiones et Approximatio Commentarii Mathematici},
  year    = {2003},
  volume  = {31},
  pages   = {23--92},
  doi     = {10.7169/facm/1538186640}
}

@book{PaleyWiener1934,
  author    = {Paley, R. E. A. C. and Wiener, Norbert},
  title     = {Fourier Transforms in the Complex Domain},
  series    = {American Mathematical Society Colloquium Publications},
  volume    = {19},
  publisher = {American Mathematical Society},
  address   = {New York},
  year      = {1934},
  doi       = {10.1090/coll/019}
}

@article{OrtmannQuastelRemenik2017,
  author  = {Ortmann, Janosch and Quastel, Jeremy and Remenik, Daniel},
  title   = {A Pfaffian Representation for Flat ASEP},
  journal = {Communications on Pure and Applied Mathematics},
  year    = {2017},
  volume  = {70},
  number  = {1},
  pages   = {3--89},
  doi     = {10.1002/cpa.21644}
}

@article{DeiftZhou1993,
  author = {Deift, Percy and Zhou, Xin},
  title = {A steepest descent method for oscillatory Riemann--Hilbert problems. Asymptotics for the MKdV equation},
  journal = {Annals of Mathematics},
  volume = {137},
  number = {2},
  pages = {295--368},
  year = {1993},
  doi = {10.2307/2946540}
}

@book{Koekoek2010,
  author = {Koekoek, Roelof and Lesky, Peter A. and Swarttouw, Ren\'e F.},
  title = {Hypergeometric Orthogonal Polynomials and Their $q$-Analogues},
  series = {Springer Monographs in Mathematics},
  publisher = {Springer},
  year = {2010},
  doi = {10.1007/978-3-642-05014-5}
}

@book{Szego1975,
  author = {Szeg\H{o}, G\'abor},
  title = {Orthogonal Polynomials},
  series = {American Mathematical Society Colloquium Publications},
  volume = {23},
  edition = {4th},
  publisher = {American Mathematical Society},
  year = {1975}
}

@article{Rains2000,
  author = {Rains, Eric M.},
  title = {Correlation functions for symmetrized increasing subsequences},
  journal = {arXiv e-print},
  eprint = {},
  year = {2000},
  url = {https://arxiv.org/abs/math/0006097}
}

@article{AkemannKanzieper2007,
  author = {Akemann, Gernot and Kanzieper, Eugene},
  title = {Integrable Structure of Ginibre's Ensemble of Real Random Matrices and a Pfaffian Integration Theorem},
  journal = {Journal of Statistical Physics},
  volume = {129},
  pages = {1159--1231},
  year = {2007},
  doi = {10.1007/s10955-007-9381-2}
}

@book{IwaniecKowalski2004,
  author = {Iwaniec, Henryk and Kowalski, Emmanuel},
  title = {Analytic Number Theory},
  series = {AMS Colloquium Publications},
  volume = {53},
  publisher = {American Mathematical Society},
  year = {2004}
}

@book{conway2012functions,
  title={Functions of one complex variable II},
  author={Conway, John B},
  volume={159},
  year={2012},
  publisher={Springer Science \& Business Media}
}

\newpage

\appendix

\section{Explicit bulk density \texorpdfstring{$\rho(u)$}{rho(u)} and microscopic spacing \texorpdfstring{$\Delta(u)$}{Delta(u)}}\label{app:densities}

%\paragraph{What this appendix contains.}
We compute the macroscopic density $\rho(u)$ and the microscopic spacing $\Delta(u)$ directly from the one–variable factors in the projection kernels. We work solely with the nested double–contour representations and the corresponding one–variable phases, with no appeal to Riemann–Hilbert analysis. Throughout we keep the contour conventions and boundedness of rational multipliers from the main text.

\subsection*{General principle (common to all three families)}
Let $K(x,y)$ denote the projection kernel in the relevant family written in its double–contour form \eqref{eq:KN-double}/\eqref{eq:KCh-double}/\eqref{eq:KNK-double}. Let $z$ be the corresponding contour variable and $\Phi(z;u)$ the one–variable phase extracted from the integrand with $x\approx \mathsf A u$. In the bulk there are two simple saddles $z_\pm(u)$ on admissible deformations of the fixed contours, with $\Re\Phi(z_\pm;u)$ equal and maximal, and $\Im\Phi(z_+;u)=-\Im\Phi(z_-;u)$.

Under the microscopic scaling
\[
x=\big\lfloor \mathsf A\,u + s\,\Delta(u)^{-1}\big\rfloor,\qquad
y=\big\lfloor \mathsf A\,u + t\,\Delta(u)^{-1}\big\rfloor,
\]
the two steepest–descent contributions acquire phases $e^{\mp \pi i s}$ and $e^{\pm \pi i t}$, respectively. The relative phase per lattice step is determined by the variation of the arguments of the saddles:
\[
\rho(u)=\frac{1}{2\pi}\,\partial_u\!\big(\arg z_+(u)-\arg z_-(u)\big),
\qquad
2\pi\,\Delta(u)\,\rho(u)=1.
\]
Thus, once $z_\pm(u)$ are solved from the saddle equation $\partial_z\Phi(z;u)=0$ (on admissible contours), $\rho(u)$ is read off from the derivative of their argument gap, and $\Delta(u)$ is recovered from the spacing rule. The bounded rational multipliers (the symbols of $D$ and $\epsilon$ in \eqref{eq:multiplier}, \eqref{eq:Ch-multipliers-w}, \eqref{eq:mult-K}) do not affect $\rho$ and $\Delta$, as they only modify finite amplitudes.

\subsection*{Meixner: explicit $\rho_M(u)$ and $\Delta_M(u)$}
From \eqref{eq:KN-double} and \eqref{eq:Gm-def}, the one–variable phase in the $\omega$–plane is
\[
\Phi_M(\omega;u)=\log\!\Big(1-\frac{s}{\omega}\Big)-\log(1-s\omega)-(u-1)\log\omega,\qquad s=\sqrt{\xi}\in(0,1).
\]
The saddle equation $\partial_\omega\Phi_M(\omega;u)=0$ simplifies to the quadratic identity
\[
u\,s\,\omega^2+\Big((1-u)-(1+u)s^2\Big)\omega+u\,s=0,
\]
with product $\omega_+(u)\omega_-(u)=1$, hence in the bulk $\omega_\pm(u)=e^{\pm i\theta(u)}$. Taking the sum gives
\[
2\cos\theta(u)=\omega_+(u)+\omega_-(u)=\frac{u(1+s^2)+(s^2-1)}{s\,u}.
\]
Therefore the bulk support is
\[
u\in\Big(\frac{1-s}{1+s},\,\frac{1+s}{1-s}\Big),
\]
and
\[
\rho_M(u)=\frac{1}{\pi}\,\theta'(u)=\frac{1-s^2}{2\pi s\,u^2}\,
\frac{1}{\sqrt{\,1-\Big(\dfrac{u(1+s^2)+(s^2-1)}{2su}\Big)^2\,}}\!,
\qquad
\Delta_M(u)=\frac{1}{2N\,\rho_M(u)}.
\]
The bounded multipliers $\widehat D,\widehat\epsilon$ in \eqref{eq:multiplier} do not affect the saddle locations and only change finite amplitudes, hence they play no role in the density.

\subsection*{Charlier: explicit $\rho_{\mathrm{Ch}}(u;\tau)$ and $\Delta_{\mathrm{Ch}}(u)$}
From \eqref{eq:KCh-double} and the generating function \eqref{eq:EGF-Ch}, the $t$–phase is
\[
\Phi_{\mathrm{Ch}}(t;u,\tau)=u\log(1+t)-\tau\,t-\log t,\qquad \tau=\lim_{N\to\infty}\frac{\theta_N}{N}\in(0,\infty).
\]
The saddle equation $\partial_t\Phi_{\mathrm{Ch}}=0$ gives
\[
\tau\,t^2+(\tau+1-u)\,t+1=0.
\]
In the bulk, the discriminant is negative, equivalent to $u\in(1+\tau-2\sqrt{\tau},\,1+\tau+2\sqrt{\tau})$. Parameterizing the saddles as $t_\pm(u)=\tau^{-1/2}e^{\pm i\theta(u)}$ yields
\[
2\cos\theta(u)=\tau^{1/2}\big(t_+(u)+t_-(u)\big)=-\frac{\tau+1-u}{\sqrt{\tau}}
\quad\Longrightarrow\quad 
\cos\theta(u)=\frac{u-(1+\tau)}{2\sqrt{\tau}}.
\]
Hence
\[
\rho_{\mathrm{Ch}}(u;\tau)=\frac{1}{\pi}\,\theta'(u)=\frac{1}{2\pi\sqrt{\tau}}\,
\frac{1}{\sqrt{\,1-\Big(\dfrac{u-(1+\tau)}{2\sqrt{\tau}}\Big)^2\,}}\!,
\qquad
\Delta_{\mathrm{Ch}}(u)=\frac{1}{N\,\rho_{\mathrm{Ch}}(u;\tau)}.
\]
\emph{Exact Meixner transfer.} Under the exact $w$–map \eqref{eq:Ch-wmap}, the multipliers become universal \eqref{eq:Ch-multipliers-w} and the kernel takes the IIKS form \eqref{eq:KCh-w}; thus the Meixner derivation above can be pulled back verbatim, but the direct $t$–plane computation already yields the same $\rho_{\mathrm{Ch}}$.

\subsection*{Krawtchouk: explicit $\rho_K(u;\gamma,p)$ and $\Delta_K(u)$}
From \eqref{eq:KNK-double} and \eqref{eq:GF-K}, the $v$–phase is
\[
\Phi_K(v;u,\gamma)=(1-u)\log(1+pv)+u\log(1-qv)-\gamma\log v,\qquad \gamma=\frac{N}{M},\quad q=1-p.
\]
The saddle equation $\partial_v\Phi_K=0$ reduces to
\[
pq(1-\gamma)\,v^2-\Big[(p-u)-\gamma(p-q)\Big]v+\gamma=0.
\]
In the bulk the roots are conjugate, which is equivalent to
\[
u\in\Big(p-\gamma(p-q)-2\sqrt{\gamma(1-\gamma)pq},\;p-\gamma(p-q)+2\sqrt{\gamma(1-\gamma)pq}\,\Big).
\]
Writing $v_\pm(u)=R\,e^{\pm i\theta(u)}$ with $R^2=\gamma/(pq(1-\gamma))$ gives
\[
2R\cos\theta(u)=v_+(u)+v_-(u)=\frac{(p-u)-\gamma(p-q)}{pq(1-\gamma)}
\quad\Longrightarrow\quad
\cos\theta(u)=\frac{(p-u)-\gamma(p-q)}{2\sqrt{\gamma(1-\gamma)pq}}.
\]
Therefore
\[
\rho_K(u;\gamma,p)=\frac{1}{\pi}\,\theta'(u)=\frac{1}{2\pi\sqrt{\gamma(1-\gamma)pq}}\,
\frac{1}{\sqrt{\,1-\Big(\dfrac{(p-u)-\gamma(p-q)}{2\sqrt{\gamma(1-\gamma)pq}}\Big)^2\,}}\!,
\qquad
\Delta_K(u)=\frac{1}{M\,\rho_K(u;\gamma,p)}.
\]

\subsection*{Summary table}
For convenience we collect the explicit formulas proved above. In each case $\Delta(u)=(\mathsf A\,\rho(u))^{-1}$ with $\mathsf A=2N$ (Meixner), $\mathsf A=N$ (Charlier), $\mathsf A=M$ (Krawtchouk).
\medskip

\renewcommand{\arraystretch}{1.2}
\begin{center}
\begin{tabular}{l|c|c|c}
\hline
Family & Bulk support $(u_-,u_+)$ & Angle $\cos\theta(u)$ & Density $\rho(u)$ \\ \hline
Meixner & $\Big(\dfrac{1-s}{1+s},\,\dfrac{1+s}{1-s}\Big)$ &
$\displaystyle \frac{u(1+s^2)+(s^2-1)}{2su}$ &
$\displaystyle \frac{1-s^2}{2\pi s u^2}\,
\frac{1}{\sqrt{1-\big(\frac{u(1+s^2)+(s^2-1)}{2su}\big)^2}}$ \\[1.1em]
Charlier & $(1+\tau-2\sqrt{\tau},\,1+\tau+2\sqrt{\tau})$ &
$\displaystyle \frac{u-(1+\tau)}{2\sqrt{\tau}}$ &
$\displaystyle \frac{1}{2\pi\sqrt{\tau}}\,
\frac{1}{\sqrt{1-\big(\frac{u-(1+\tau)}{2\sqrt{\tau}}\big)^2}}$ \\[1.1em]
Krawtchouk & $p-\gamma(p-q)\pm 2\sqrt{\gamma(1-\gamma)pq}$ &
$\displaystyle \frac{(p-u)-\gamma(p-q)}{2\sqrt{\gamma(1-\gamma)pq}}$ &
$\displaystyle \frac{1}{2\pi\sqrt{\gamma(1-\gamma)pq}}\,
\frac{1}{\sqrt{1-\big(\frac{(p-u)-\gamma(p-q)}{2\sqrt{\gamma(1-\gamma)pq}}\big)^2}}$ \\
\hline
\end{tabular}
\end{center}

\medskip
\noindent\textbf{Remarks.}
(i) At either endpoint $u_\pm$ of the bulk support, $\rho(u)\sim C_\pm\,\sqrt{|u-u_\pm|}$, matching the soft–edge cubic reduction used in the Airy limits. 
(ii) The bounded multipliers (Meixner \eqref{eq:multiplier}; Charlier \eqref{eq:Dhat-Ch-t}–\eqref{eq:eps-hat-Ch-t}; Krawtchouk \eqref{eq:mult-K}) do not change $\rho(u)$: they shift only finite amplitudes and thus contribute at subleading orders in the steepest–descent evaluation.

\section{Uniform steepest–descent estimates}\label{app:steepest}

This appendix records the standard local/outer split for the one–variable contour integrals used in Section~\ref{sec:asymp-univ-long}, with bounds uniform for $u$ in compact bulk sets and for soft–edge windows; see also \cite{DeiftZhou1993} for the general RHP steepest–descent framework. We choose local coordinates so that the phase function has the form $e^{i N \theta(z)}$ with $\theta'(z_{\pm})=0$ at saddles $z_{\pm}$, then expand $\theta(z)$ to second order to get a Gaussian integral; by ensuring no other stationary points interfere, one controls the error $O(N^{-1})$ uniformly.

\subsection*{Bulk: uniform Gaussian reduction}
Fix a bulk $u$ and let $z_\pm(u)$ be the two simple saddles on admissible steepest–descent arcs. Introduce local charts $\zeta_\pm$ by $z=z_\pm e^{\zeta_\pm}$ on small sectors containing the steepest directions. Then
\[
\Phi(z;u)=\Phi(z_\pm;u)+\tfrac12\Phi''(z_\pm;u)\,\zeta_\pm^2+R_\pm(\zeta_\pm;u),\qquad |R_\pm(\zeta_\pm;u)|\le C|\zeta_\pm|^3
\]
for $|\zeta_\pm|\le \zeta_0$, with $C,\zeta_0$ independent of $u$ in compact bulk sets. Split each contour into the ``local'' piece $|\zeta_\pm|\le \mathsf A^{-1/3}$ and its complement:
\begin{itemize}
\item On the complement, $\Re(\Phi-\Phi(z_\pm))\le -c|\zeta_\pm|^2$ along steepest descent, hence the contribution is $O(e^{-c\mathsf A^{1/3}})$.
\item On the local piece, replace the integrand by its quadratic Taylor expansion and evaluate the Gaussian exactly. Shifts $x\mapsto x+s\,\Delta^{-1}$ and $y\mapsto y+t\,\Delta^{-1}$ contribute phase factors $e^{\mp\pi i s}$, $e^{\pm\pi i t}$; the spacing rule $2\pi\Delta\rho=1$ enforces the $2\pi$ phase change per lattice step.
\end{itemize}
This proves Lemma~\ref{lem:gauss} with the uniform window $|s|,|t|\le \mathsf A^\delta$ for any $\delta<\tfrac12$ and the $O(\mathsf A^{-1})$ error.

\subsection*{Soft edge: cubic/Airy normal form and the scaling constant}
At a soft edge $(z_\ast,u_\ast)$ with $\Phi'(z_\ast;u_\ast)=\Phi''(z_\ast;u_\ast)=0$, write
\[
\Phi(z;u)=\Phi(z_\ast;u_\ast)+\tfrac{\kappa}{3}\zeta^3-\eta\,\lambda\,\zeta+O(\zeta^4)+O(\eta\zeta^2),
\]
where $\zeta$ is a local edge chart and $\eta$ is the $u$–offset. Choosing $x=\lfloor \mathsf A u + s\,c\,\mathsf A^{1/3}\rfloor$ with $c=(\kappa/\lambda)^{1/3}$ reduces the one–variable integrals to Airy integrals with errors $O(\mathsf A^{-1/3})$ uniformly in fixed soft–edge windows, yielding Theorem~\ref{thm:airy-edge}.

%%%%%%%%%%%%%%%%%%%%%%%%%%%%%%%%%%%%%%%%%%%%%%%%%%%%%%%%%%%%%%%%%%%%%%%%%%%%

\end{document}